\documentclass[journal]{IEEEtran}
\usepackage{amsmath}
\usepackage{amsthm}
\usepackage{amsfonts}
\usepackage{mathrsfs}
\usepackage{pifont}
\usepackage{amssymb}
\usepackage{verbatim}
\usepackage{upgreek}
\usepackage{color}
\usepackage[final]{graphicx}
\usepackage{algorithmic}
\usepackage{algorithm}
\usepackage{epsfig}
\usepackage{wrapfig}
\usepackage{psfrag}
\usepackage{subfigure}
\usepackage{cite}
\usepackage{mathtools}

\newcommand{\bzero}{\mbox{\boldmath{$0$}}}

\newcommand{\bt}{\mbox{\boldmath{$t$}}}
\newcommand{\bU}{\mbox{\boldmath{$U$}}}
\newcommand{\bu}{\mbox{\boldmath{$u$}}}

\newcommand{\bx}{\mbox{\boldmath{$x$}}}

\newcommand{\by}{\mbox{\boldmath{$y$}}}

\newcommand{\bphi}{\mbox{\boldmath{$\phi$}}}
\newcommand{\bomega}{\mbox{\boldmath{$\omega$}}}
\newcommand{\bzeta}{\mbox{\boldmath{$\zeta$}}}

\newcommand{\bnu}{\mbox{\boldmath{$\nu$}}}

\newtheorem{theorem}{Theorem}[section]

\newtheorem{lemma}[theorem]{Lemma}

\graphicspath{ {Figures/} }
\begin{document}

\title{A Coordinate-Descent Framework to Design Low PSL/ISL Sequences}

\author{M. Alaee, \thanks{M. Alaee, {M.} {M.} Naghsh and {M.} Modarres-Hashemi are  with the Department of Electrical and Computer Engineering, Isfahan University of Technology, Isfahan 84156-83111, Iran. Email: m.alaee@ec.iut.ac.ir, mm\_naghsh@cc.iut.ac.ir, modarres@cc.iut.ac.ir}\emph{Student Member, IEEE}, A. Aubry, \emph{Senior Member, IEEE}, A. De Maio\thanks{
A. Aubry and A. De Maio are with
the Universit\`a
degli Studi di Napoli ``Federico II'', Dipartimento di Ingegneria
Elettrica e delle Tecnologie dell'Informazione, Via Claudio 21,
I-80125 Napoli, Italy. E-mail: augusto.aubry@unina.it, ademaio@unina.it.}, \emph{Fellow, IEEE}, {M.} {M.} Naghsh, \emph{Member, IEEE}, M. Modarres-Hashemi}

\maketitle

\begin{abstract}
 This paper is focused on the design of phase sequences with good (aperiodic) autocorrelation properties in terms of Peak Sidelobe Level (PSL) and Integrated Sidelobe Level (ISL). The  problem is formulated as a bi-objective Pareto optimization forcing either a continuous or a discrete phase constraint at the design stage. An iterative procedure based on the coordinate descent method is introduced to deal with the resulting optimization problems which are non-convex and NP-hard in general.
Each iteration of the devised method requires the solution of a non-convex min-max problem. It is handled either through a novel bisection or an FFT-based method for the continuous and the discrete phase constraint, respectively. Additionally, a heuristic approach to initialize the procedures employing the $l_p$-norm minimization technique is  proposed. Simulation results illustrate that the proposed methodologies can outperform some counterparts providing sequences with good autocorrelation features especially in the discrete phase/binary case.
\end{abstract}

\section{Introduction}
Waveform design has received considerable attention during the last eight decades in many communication, active sensing, and electronic warfare systems \cite{woodward1953probability}.
In communication systems (e.g., code-division multiple access), low correlation sidelobes are desired for synchronization and reduction of multi-access interferences \cite{9781139095174}. In radar range compression, low Peak Sidelobe Level  (PSL) waveforms are employed to avoid masking of weak targets in the range sidelobes of a strong return \cite{barker1953group,68151}. Besides, to mitigate the deleterious effects of distributed clutter returns close to the target of interest \cite{4517015}, signals with low Integrated Sidelobe Level (ISL) are exploited. Remarkably, coded waveforms and range compression grant enhanced electronic protection against barrage jamming 
 as well as  increased range resolution which is critical in electronic warfare, e.g., to prevent  deceptive attacks based on range-gate pull-off  \cite{jammingAntonio}. Also, pulse to pulse changing the transmitted waveform totally counters range-gate pull-in attacks \cite{jammingAntonio}. Generally speaking, according to the hardware technology involved in the waveform generation process, the signals employed in real systems can be classified in two types: analog and digital. The former class exploits analog circuits/devices and typical examples are the Linear Frequency Modulated (LFM) and the non-linear frequency modulated signals; the latter is based on arbitrary digital waveform generators and relevant instances are the phase coded waveforms involving Barker, Frank, and Golomb, sequences, just to list a few \cite{levanon2004radar}. Both conventional phase coded and LFM waveforms are ubiquitous in practical systems mainly due to their easy generation as well as the possible tolerance to Doppler-shifts.
 Nevertheless, the static use of a fixed waveform could determine a scarce adaptivity to the operating environment as well as  vulnerability to electronic attacks highlighting the need for multiple and diverse waveforms exhibiting specific features \cite{DeMaiobook}. As a result, several researchers have proposed a variety of approaches based on sophisticated optimization methods to design advanced polyphase sequences \cite{5072243,6601713,6650043,7362231}. This paper is framed in the mentioned context with the goal of designing constant-modulus sequences whose autocorrelation sidelobes share a desirable behavior.
\subsection{Background and Previous Works}
Binary phase coded waveforms are common in radar and communication systems being their implementation quite simple \cite{skolnik2008radar}. Important instances are the Barker codes  which are unfortunately limited  to length 13.  $M$-sequences, well-known for their ideal periodic autocorrelation function, can  be easily generated using linear feedback shift registers but have no constraints/guarantees on the sidelobes of their aperiodic autocorrelation function; hence, they are almost impractical in radar applications (similarly Gold codes and Kasami sequences)\footnote{Finding sequences with good aperiodic correlation properties is usually a harder task than searching for sequences
with good periodic correlation.}. Unlike the
case of periodic correlation, it is not possible to construct binary sequences
with an exact impulsive aperiodic autocorrelation. Therefore, a brute-force approach to obtain good sequences is to perform an exhaustive search, viable especially  when the alphabet size is small, i.e., binary case. Minimum Peak Sidelobe (MPS) sequences are the best binary codes in terms of PSL (known up to length 105) which are obtained via global search through some supercomputers; a summary of the best known binary sequences is presented in \cite{nasrabadi2010survey}. When the constellation size increases, it becomes difficult and difficult (almost impossible) to perform the exhaustive search.
In these situations derivation of analytical methods to design optimal or nearly-optimal sequences are valuable. To this end,
 an Iterative Twisted appROXimation (ITROX) method is proposed in \cite{6142119} to get discrete phase sub-optimal sequences with limited alphabet size. Heuristic techniques exploiting Simulated Annealing (SA), Threshold Accepting (TA), and Great Deluge Algorithm (GDA) have been developed to design generalized polyphase Barker sequences\footnote{A polyphase sequence is a generalized Barker code if the magnitude of all the autocorrelation sidelobes is less than or equal to one but for the last entry that is one \cite{508850}.} \cite{508850,1412048,1435817}. However, the maximum length of the obtained generalized Barker codes is limited to $77$ \cite{5089560}.\\
In \cite{4567663} a computationally efficient cyclic optimization algorithm for the
design of constant-modulus transmit signals with good auto- and cross-correlation features is developed with reference to Multi-Input-Multi-Output (MIMO) radars. Following a similar line of reasoning, in \cite{4749273} and \cite{5072243} cyclic algorithms for the  minimization of ISL-related metrics in Single-Input-Single-Output (SISO) systems are introduced, i.e., Cyclic Algorithm Pruned (CAP), Cyclic Algorithm New (CAN), and Weighted CAN (WeCAN). The main merit of these procedures is the reduced computational complexity that leads to a quite short execution time. From a theoretical point of view, they are based on the solution of an optimization problem that is \textit{almost equivalent} to the ISL minimization, being the objective function an approximation of the ISL.
Recently, a new optimization algorithm which can directly minimize the ISL metric under a constant modulus constraint\footnote{In \cite{6563125} a general framework (with theoretical ensured convergence properties) to optimize quartic order functions assuming phase-only sequences is presented. Interestingly, it can also be used to design optimized ISL sequences.} has been introduced in \cite{7362231}. It relies on the Majorization-Minimization (MM) framework and generally converges to a locally optimum ISL value. Two different versions are proposed exploiting two distinct majorization functions (Majorization-minimization Weighted ISL (MWISL) and MWISL-Diag). Their performance in terms of achieved ISL values is similar to that of CAN algorithm. Meanwhile, adopting a suitable modification in both MWISL and MWISL-Diag (as specified in \cite{7362231}), the aforementioned techniques can converge faster than CAN.\\
 Most of the above literature is focused on ISL-oriented design problems mainly due to the technical difficulties arising when the non-differentiable and highly non-convex PSL metric is considered as figure of merit. A first attempt to fill this gap and systematically synthesize polyphase codes with low PSL values is pursued in \cite{7362231}. Precisely, the authors provide an algorithm, based on MM paradigm, capable of minimizing the $l_p$-norm of the autocorrelation sidelobes. Hence, a heuristic method (called MM-PSL) exploiting the observation that the PSL coincides with the
  limit as $p$ goes to infinity of the mentioned $l_p$-norm\footnote{A similar approach has been also investigated in \cite{5765722} to devise optimized receive filters.} is proposed to optimize the PSL.
 \subsection{Contribution and Organization}
In this paper,  PSL and ISL, namely the two most important measures quantifying the quality of the autocorrelation function, are jointly considered to synthesize advanced constant modulus codes according to a Pareto optimization framework.
Specifically, the problem is formulated as a bi-objective optimization where either a continuous or a
discrete phase constraint is forced at the design stage. To tackle the resulting non-convex and, in general, NP-hard problems
an iterative procedure based on the Coordinate Descent (CD) method is introduced. Each iteration of the developed procedure requires the solution of a non-convex min-max problem involving quartic functions.
As to the continuous phase case, a novel polynomial-time bisection method aimed at solving globally the aforementioned  problem is developed. The discrete phase design, encompassing the challenging binary synthesis, is handled via an FFT-based procedure. Besides, a heuristic approach to initialize the procedures exploiting an $l_p$-norm minimization criterion is introduced.

Summarizing, the contributions of this paper are:
\begin{itemize}
\item the development of an efficient CD method that optimizes an objective function given by a weighted sum of the ISL and PSL based metrics. The method decreases the value of the objective  at each iteration and can ensure convergence to a stationary point provided that the Maximum Block Improvement (MBI) \cite{chen2012maximum,6563125} rule is adopted. Also, the complexity  per iteration is polynomial.  Remarkably, to the best of our knowledge, no mathematical algorithm with ensured convergence properties has been suggested in the literature for the exact PSL minimization.
\item the design of sequences with discrete phase possessing low ISL and PSL. In this respect, it is worth observing that there exist algorithms in open literature for ISL minimization. Nevertheless, they usually do not perform well in the discrete phase case \cite{4567663} and the proposed  algorithm is able to outperform them. As to the PSL minimization, to the best of our knowledge, systematic approaches are not available in the open literature  and our method fills this relevant gap.
\item the specialization of the proposed design methodology to the context of binary sequences with good ISL and PSL. Precisely, at each iteration of the devised method a weighted sum of the ISL
and the PSL based metrics of the starting binary code decreases until convergence.
\end{itemize}
The rest of this paper is organized as follows. Section \ref{system_model} deals with  problem formulation. In Section \ref{CD-based}, the CD-based solution method is introduced together with techniques aimed at solving the optimization problem involved in each iteration, for both the continuous and the discrete phase case. Besides, a heuristic method is discussed to initialize the new proposed algorithms.
 Numerical examples are provided in Section \ref{perf_analysis} to illustrate the effectiveness of the approach. Finally, concluding remarks  and possible future research tracks are given in Section \ref{conclusions}.

\subsection{Notation}
We adopt the notation of using bold lowercase letters for vectors and bold uppercase letters  for matrices.  The transpose, the conjugate, and
the conjugate transpose operators are denoted by the
symbols ${(\cdot)^T}$, $(\cdot)^*$ and $(\cdot)^H$ respectively; The $l_p$-norm of a vector $\bx$ is denoted by $\|\bx\|_p$. The letter $\jmath$ represents the imaginary unit (i.e., $\jmath = \sqrt{-1}$), while the letter $i$ often serves as index. For any complex number
$x$, we use $\Re(x)$ and $\Im(x)$ to denote the real part and the imaginary part of $x$, respectively. For any $x\in \mathbb{R}$, $\lceil x\rceil$ denotes the lowest integer higher than or equal to $x$. Also, $|x|$ and $\arg(x)$ represent the modulus and the argument of $x$, respectively. The abbreviation ``s.t.'' stands for ``subject to''.

\section{Problem Formulation}\label{system_model}
Let $\bx = [x_1, x_2, \ldots, x_N]^T \in {\mathbb{C}}^N$ be the transmitted fast-time radar code vector with $N$ the number of coded sub-pulses (code length). The autocorrelation function associated with $\bx$ is defined as
 \begin{equation}\label{rk1}
r_k = \sum_{i=1}^{N-k}{x_i^*x_{i+k}}, ~~~~ k = 0, \ldots , N-1,
\end{equation}
and represents the output of the matched filter to $\bx$ when $\bx$ is the input signal. The PSL and
ISL metrics\footnote{Notice that in some references  $\text{ISL} = 2 \sum_{k=1}^{N-1}|r_k|^2$.} are commonly used to design waveforms with ``\emph{good}'' autocorrelation properties \cite{levanon2004radar} and are
formally defined as
\begin{eqnarray}
\text{PSL} &=&  \max\{|r_k|\}_{k=1}^{k=N-1},\\ \label{psl1}
\text{ISL} &=&  \sum_{k=1}^{N-1}|r_k|^2,\label{isl1}
\end{eqnarray}
respectively. This paper is focused on the design of unimodular sequences considering simultaneously the PSL and the ISL as performance indices. From an analytical point of view the problem can be formulated as the following
constrained bi-objective optimizations,
\begin{flalign}\label{eq:multi1}
P^{\infty}
\begin{dcases}
\min_{\bx}& {f_1(\bx), f_2(\bx)}\\
s.t. & \bx \in \Omega_{\infty}
\end{dcases}
&,~~
P^{M}
\begin{dcases}
\min_{\bx}& {f_1(\bx), f_2(\bx)}\\
s.t. & \bx \in \Omega_{M}
\end{dcases}
\end{flalign}
where
\begin{equation*}
f_1{(\bx)} = \max\{|r_k|^2\}_{k=1}^{k=N-1}
\end{equation*}
and
\begin{equation*}
f_2{(\bx)} = \sum_{k=1}^{k=N-1}\{|r_k|^2\}.
\end{equation*}
Herein, the constraints $\bx \in \Omega_{\infty}$ and $\bx \in \Omega_M$ denote continuous alphabet\footnote{Continuous alphabet means that there is no constraint on the phase values which can get any arbitrary value within ${[-\pi, \pi]}$. In the case of discrete phase constraint, the feasible set is restricted to a finite number of equi-spaced points on the unit circle.} and finite alphabet codes, respectively. Precisely,
\begin{equation}
\Omega_{\infty} = \{\bx \in {\mathbb{C}}^{N} | ~|x_i| = 1, i = 1, \ldots, N\}
\end{equation}
and
\begin{equation}
 \Omega_M = \{\bx  |  x_i \in \Psi_M , i = 1, \ldots, N\},
 \end{equation}
 where $\Psi_M = \{1, \bar{\omega}, \ldots, \bar{\omega}^{M-1}\}$, $ \bar{\omega} = e^{\jmath \frac{2\pi}{M}}$ and $M$ is the size of discrete constellation alphabet.\\
In a multi-objective optimization framework, usually a feasible solution that minimizes all the objective functions simultaneously does not exist \cite{deb2001multi}.
Accordingly, the goal is to find the \textit{Pareto-optimal} solutions to \eqref{eq:multi1} which is in general a formidable task.
A viable means to obtain the above solutions is the \textit{scalarization technique}\footnote{Scalarizing a multi-objective problem involves the solution of conventional optimization problems whose objective function is a specific convex combination of the original figures of merits \cite{boyd2004convex}. One or more Pareto-optimal solutions  correspond  to each selected weight vector.} which exploits  as objective a specific weighted sum  between $f_1{(\bx)}$ and $f_2{(\bx)}$. Specifically, defining the function $f_{\theta}(\bx)$, parameterized in the Pareto weight  $\theta \in [0, 1]$,
\begin{align} \label{eq:fbx}
f_{\theta}(\bx) = & ~ {\theta f_1(\bx) + (1-\theta) f_2(\bx)} \nonumber \\
= & \displaystyle{\max_{k=1,\ldots, N-1}}\left[  \theta {|r_k|^2}+ ( 1- \theta ) \sum_{l=1}^{N-1}{|r_l|^2} \right]
\end{align}
scalarization leads to the design problems
 \begin{flalign}\label{eq:Pareto25}
P^{\infty,\theta}
\begin{dcases}
\min_{\bx}& f_{\theta}(\bx)\\
s.t. & \bx \in \Omega_{\infty}
\end{dcases}
&,
P^{M,\theta}
\begin{dcases}
\min_{\bx}& f_{\theta}(\bx)\\
s.t. & \bx \in \Omega_{M}
\end{dcases}
\end{flalign}
They reduce to pure ISL (PSL) minimization setting $\theta = 0$ ($\theta = 1$).
Moreover, for any $\theta$, an optimal solution to \eqref{eq:Pareto25} is a Pareto-optimal point to Problem \eqref{eq:multi1}
(see {\cite{de2011pareto,de2011design,boyd2004convex}} and references therein for details).
\section{CD Code Optimization}\label{CD-based}
This section introduces an iterative algorithm based on the CD minimization procedure \cite{wright2015coordinate} (also known as alternate optimization \cite{1232330}) to sequentially optimize our objective over one variable keeping fixed the others. Otherwise stated, according to the CD approach, the minimization of a multivariable function can be achieved minimizing it along one direction at a time, i.e., solving univariate optimization problems in a loop \cite{wright2015coordinate, Richtárik2014}. With reference to (\ref{eq:Pareto25}), at each iteration a specific code entry is selected as variable to optimize leading to the following problems at step $n+1$
 \begin{flalign}\nonumber
P^{\infty,\theta}_{d,\bx^{(n)}}
\begin{dcases}
\min_{x_d}&f_{\theta}(x_d;\bx_{-d}^{(n)})\\
s.t. & |x_d|  = 1
\end{dcases}
&,
P^{M,\theta}_{d,\bx^{(n)}}
\begin{dcases}
\min_{x_d}& f_{\theta}(x_d;\bx_{-d}^{(n)}) \\
s.t. & x_d \in \Omega_{M}
\end{dcases}
\end{flalign}
 where $x_d$ is the variable to optimize, ${\bx_{-d}^{(n)} = [x_1^{(n)},\ldots,x_{d-1}^{(n)},x_{d+1}^{(n)},\ldots,x_{N}^{(n)}]^T \in {\mathbb{C}}^{N-1}}$ refers to the remaining code entries, and
 \begin{equation*}
f_{\theta}(x_d;\bx_{-d}^{(n)}) = f_{\theta}(x_1^{(n)},\ldots,x_{d-1}^{(n)},x_d,x_{d+1}^{(n)},\ldots,x_{N}^{(n)}).
\end{equation*}
 Thus, denoting by $x_{d,n+1}^{\star}$ the optimal solution to either $P^{\infty,\theta}_{d,\bx^{(n)}}$ or $P^{M,\theta}_{d,\bx^{(n)}}$, the optimized radar code at step $n+1$ is $\bx^{(n+1)}=[x_1^{(n)},\ldots,x_{d-1}^{(n)},x_{d,n+1}^{\star},x_{d+1}^{(n)},\ldots,x_{N}^{(n)}]^T$.  As a result, starting from an initial code $\bx^{(0)}$ a sequence $\bx^{(1)},\bx^{(2)},\bx^{(3)},\ldots$  of radar codes are obtained iteratively. A summary of the proposed approach can be found in {\bf Algorithm \ref{alg_seq}}.\\
\begin{algorithm}
\caption{Continuous (Discrete) Phase Code Design with Good Autocorrelation Features}
\label{alg_seq}
\textbf{Input:} Initial code $\bx_0 \in \Omega_{\infty}$ ($\bx_0 \in \Omega_{M}$), $\theta \in [0, 1]$, and minimum required improvement $\epsilon$; \\
\textbf{Output:} Optimal solution $\bx^{\star}$;
\begin{enumerate}
\item {\bf Initialization}.
\begin{itemize}
\item Compute the initial objective value ${f_{\theta}(x^{(0)}_1,x^{(0)}_2,\ldots,x^{(0)}_N)}$ using equation \eqref{eq:fbx};
\item  Set $d := 1$ and $n := 0$;
\end{itemize}
 \item {\bf Improvement}.
\begin{itemize}
\item Solve $P^{\infty,\theta}_{d,\bx^{(n)}}$ ($P^{M,\theta}_{d,\bx^{(n)}}$) obtaining $x_d^{\star}$;
\item Set ${n := n+1}$ and \\${\bx^{(n)}\! =\! \left[x_1^{(n-1)},\ldots,x_{d-1}^{(n-1)}\!,\!x_d^{\star},x_{d+1}^{(n-1)}\!,\ldots,x_{N}^{(n-1)}\!\right]^T}\!$;
\end{itemize}
\item {\bf Stopping Criterion}.
\begin{itemize}
\item {If} $|f_{\theta}(\bx^{(n)}) - f_{\theta}(\bx^{(n-1)})| \ < \epsilon$, stop. Otherwise, update $d$, i.e.,  if $d < N$ $d=d+1$, otherwise $d = 1$, and go to the step 2;
\end{itemize}
\item {\bf Output}.
\begin{itemize}
\item Set $\bx^{\star} = \bx^{(n)}$.
\end{itemize}
\end{enumerate}
\end{algorithm}
Notice that, the monotonic property of the CD technique along with the fact that the objective function is bounded (from below) are sufficient to prove the convergence of the sequence of objective values. It is also worth pointing out that the Maximum Block Improvement (MBI) updating rule\footnote{The MBI method is an iterative algorithm known to achieve excellent performance
in the maximization of real polynomial functions subject to spherical constraints \cite{6563125}. It is proved that any cluster point of the sequence produced by the MBI method is a stationary point for the considered optimization problem\cite{chen2012maximum}.} \cite{chen2012maximum} can be used in place of the  cyclic one (actually involved in {\bf Algorithm 1}) to ensure the convergence of the algorithm to a stationary point. In practice, a final optimized code can be obtained refining the solution provided {\bf Algorithm 1} through the MBI-modification.\\
To proceed further, let us make explicit the functional dependence of the objective function in  $P^{\infty,\theta}_{d,\bx^{(n)}}$  ($P^{M,\theta}_{d,\bx^{(n)}}$), i. e., $f_{\theta}(x_d;\bx_{-d}^{(n)})$ over the optimization variable $x_d$, i.e.,
\begin{equation}
\begin{aligned}
r_{k}(x_d)  = & x_d x_{d+k}^{*}\mathbf{1}_A(d+k) + x_{d-k}x_d^*\mathbf{1}_A(d-k) \\
& + \sum_{i=1 , i \neq\{ d, d-k\}}^{N-k}{x_i x_{i+k}^*}, ~ k = 1 , \ldots , N-1,
\end{aligned}
\end{equation}
where $\mathbf{1}_A(.)$ denotes the
indicator function of the set $A = \{1,2,\ldots,N\}$, i.e., $\mathbf{1}_A(x) = 1$ if $x \in A$, otherwise $\mathbf{1}_A(x) = 0$.
Defining, $a_{dk} \triangleq x_{d+k}^{*}\mathbf{1}_A(d+k)$, $b_{dk} \triangleq x_{d-k}\mathbf{1}_A(d-k) $ and $c_{dk}  \triangleq \sum_{i=1 , i \neq \{d , d - k\}}^{N-k}{x_i x_{i+k}^*}$, the autocorrelation function with the explicit $x_d$-dependence can be written as
\begin{equation}
r_{k}(x_d) = a_{dk}x_d +b_{dk}x_d^*+ c_{dk}  , ~~~ k = 1 , \ldots , N-1.
\end{equation}
Thus, the optimization problems $P^{\infty,\theta}_{d,\bx^{(n)}}$  and $P^{M,\theta}_{d,\bx^{(n)}}$  can be recast as,
\begin{flalign*}
P^{\infty,\theta}_{d,\bx^{(n)}}
\begin{dcases}
\min_{x_d}\displaystyle{\max_{k=1,\ldots, N-1}}&\!\!\!\!\left[  \theta { \left|{r_k(x_d)} \right|^2}+ ( 1- \theta ) \!\!\sum_{l=1}^{N-1}{ \left|{r_l(x_d)} \right|^2} \right]\\
s.t. & |x_d|  = 1
\end{dcases}
\end{flalign*}
\begin{flalign*}
P^{M,\theta}_{d,\bx^{(n)}}
\begin{dcases}
\min_{x_d}\displaystyle{\max_{k=1,\ldots, N-1}}&\!\!\!\!\left[  \theta { \left|{r_k(x_d)} \right|^2}+ ( 1- \theta )\!\! \sum_{l=1}^{N-1}{ \left|{r_l(x_d)} \right|^2} \right]\\
s.t. & x_d \in \left\{1, e^{\jmath\frac{2\pi}{M}},\ldots,e^{\jmath\frac{2\pi(M-1)}{M}}\right\}
\end{dcases}
\end{flalign*}
which are non-convex, constrained, min-max problems with non-homogeneous quadratic objectives of a complex variable.
In the next subsections, efficient algorithms to tackle  $P^{\infty,\theta}_{d,\bx^{(n)}}$ and $P^{M,\theta}_{d,\bx^{(n)}}$ are derived. This issue represents the main technical innovation of this paper from an optimization theory point of view.
\subsection{Continuous Phase Code Design}
This subsection is focused on the solution of Problem $P^{\infty,\theta}_{d,\bx^{(n)}}$. As first step toward this goal, it is shown that the square modulus of the autocorrelation function at each lag (as a function of $x_d= e^{\jmath \phi_d}, \phi_d \in [0, 2\pi]$) can be expressed as the the ratio of two quartic functions of a real variable. This result is given in terms of the following
\begin{lemma}\label{lemm_poly1} Performing the change of variable $\beta_d \triangleq \tan\left({\frac{\phi_d}{2}}\right)$,
\begin{equation} \label{eq:rk2}
\begin{aligned}
|\widetilde{r}_{k}(\beta_d)|^2 & =  |r_{k}(e^{\jmath\phi_{d}})|^2 \\
& = \frac{\mu_{dk} \beta_d^4 + \kappa_{dk} \beta_d^3 + \xi_{dk} \beta_d^2 + \eta_{dk} \beta_d + \rho_{dk}}{(1+\beta_d^2)^2},
\end{aligned}
\end{equation}
where $\mu_{dk}, \kappa_{dk} , \xi_{dk}, \eta_{dk}, \rho_{dk}$ are real-valued coefficients depending on $a_{dk}$, $b_{dk}$ and $c_{dk}$ as specified in Appendix \ref{ap:rk2}.\\
\end{lemma}
\begin{proof}
see Appendix \ref{ap:rk2}.
\end{proof}
Based on Lemma \ref{lemm_poly1}, Problem $P^{\infty,\theta}_{d,\bx^{(n)}}$  is equivalent to the following optimization problem,
\begin{flalign}\label{eq:ParetoB111}
\bar{P}^{\infty,\theta}_{d,\beta_d}
\begin{dcases}
\min_{\beta_d\in {\cal{R}}} &\displaystyle{\max_{k=1,\ldots, N-1}}\frac{\widetilde{p}_{k}(\beta_d)}{\widetilde{q}(\beta_d)}
\end{dcases}
\end{flalign}
where
\begin{equation}
\widetilde{p}_{k}(\beta_d) = {\theta {p}_{k}(\beta_d) + (1 - \theta) \sum_{l=1}^{N-1}{p}_{l}(\beta_d)}
\end{equation}
 and
 \begin{equation}\label{eq:qkbeta}
 \widetilde{q}(\beta_d) = (1+\beta_d^2)^2,
 \end{equation}
with
\begin{equation}\label{eq:pkbeta}
{p}_{k}(\beta_d) = \mu_{dk} \beta_d^4 + \kappa_{dk} \beta_d^3 + \xi_{dk} \beta_d^2 + \eta_{dk} \beta_d + \rho_{dk}.
\end{equation}
In particular, ${\widetilde{p}_{k}(\beta_d)}$ and ${\widetilde{q}(\beta_d)}$ are  non-negative quartic polynomials. Now, let $\bar{\gamma} \in {\cal{R}}^{+}$ be an slack variable and $v^{\star}$ the optimal value of the min-max optimization Problem $\bar{P}^{\infty,\theta}_{d,\beta_d}$  whose existence is ensured by Weierstrass theorem applied to $P^{\infty,\theta}_{d,\bx^{(n)}}$. It can be checked whether the optimal value $v^{\star}$ is lower than or equal to a given value $\bar{\gamma}$ solving the feasibility problem
\begin{flalign}\label{eq:optfeas}
\widetilde{P}^{\infty,\theta}_{\beta_d,\bar{\gamma}}
\begin{dcases}
\text{find} & \beta_d\\
s.t. & \frac{\widetilde{p}_{k}(\beta_d)}{\widetilde{q}(\beta_d)} \leq {\bar{\gamma}}, ~~ k = 1, \ldots, N-1
\end{dcases}
\end{flalign}
If $\widetilde{P}^{\infty,\theta}_{\beta_d,\bar{\gamma}}$ is feasible, then $v^{\star} \leq \bar{\gamma}$ and there exists a point in the feasible set achieving an objective value better than or equal to $\bar{\gamma}$. Conversely, if  Problem $\widetilde{P}^{\infty,\theta}_{\beta_d,\bar{\gamma}}$ is infeasible $v^{\star} \geq \bar{\gamma}$.
The above observation paves the way  to the development of an efficient iterative algorithm to solve $\bar{P}^{\infty,\theta}_{d,\beta_d}$ according to the bisection approach \cite{boyd2004convex}.  Precisely, at step $i$ the feasibility Problem \eqref{eq:optfeas} is solved with $\bar{\gamma} = \frac{u_{(i)}+w_{(i)}}{2}$ where $[w_{(i)},u_{(i)}]$  is the current interval containing the optimal value\footnote{As starting interval, $w_{(0)} = 0$ and $u_{(0)} =  f_{\theta}\left(\bx^{(n)}\right)$  are considered.} $v^{\star}$. Based on the feasibility check, it is possible to determine whether the optimal value is in the lower or in the upper half of the current interval and update the search accordingly with a consequent uncertainty halving. The procedure is repeated until the width of the interval is lower than or equal to a prescribed accuracy\footnote{The proposed algorithm can be used to solve any arbitrary generalized fractional programming problem involving quartic order functions of a real variable, with strictly positive denominators.}. \\
To study the feasibility of $\widetilde{P}^{\infty,\theta}_{d,\beta_d}$ for a given $\bar{\gamma}$, let us define the feasible set $ {\cal{A}}_k^{ \bar{\gamma}}$ as,
\begin{equation} \label{eq:calAk}
 {\cal{A}}_k^{ \bar{\gamma}} = \left\{ \beta_d | \left[ \widetilde{p}_{k}(\beta_d) - \bar{\gamma}\widetilde{q}(\beta_d)\right]\leq 0 \right\} , k = 1, \ldots, N-1,
 \end{equation}
 and the complement set $ \overline{\cal{A}}_k^{ \bar{\gamma}}$ as,
 \begin{equation}\label{eq:calAkbar}
 \overline{\cal{A}}_k^{ \bar{\gamma}} = \left\{ \beta_d | \left[ \widetilde{p}_{k}(\beta_d) - \bar{\gamma}\widetilde{q}(\beta_d)\right] > 0 \right\},  k = 1, \ldots, N-1.
 \end{equation}
 Therefore $\widetilde{P}^{\infty,\theta}_{d,\beta_d}$ is feasible if and only if,
 \begin{equation} \label{eq:feas_set}
 \bigcap_{k = 1}^{N-1}{\cal{A}}_k^{ \bar{\gamma}} \neq \varnothing
  \Leftrightarrow \overline{\bigcup_{k = 1}^{N-1}\overline{\cal{A}}_k^{ \bar{\gamma}}}\neq \varnothing
\Leftrightarrow {\bigcup_{k = 1}^{N-1}\overline{\cal{A}}_k^{ \bar{\gamma}}}\neq {\cal{R}}.
 \end{equation}
Conversely, if ${\cup_{k = 1}^{N-1}\overline{\cal{A}}_k^{ \bar{\gamma}}}= {\cal{R}}$ Problem $\widetilde{P}^{\infty,\theta}_{d,\beta_d}$ is infeasible. In a nutshell, to perform the feasibility check it is enough to compute the union of the intervals $\overline{{\cal{A}}}_k^{ \bar{\gamma}}$ defined in \eqref{eq:calAkbar} and to check for the possible gaps.
In this respect, an efficient technique to establish the presence of gaps can be developed just finding the roots of $\widetilde{p}_{k}(\beta_d) - \bar{\gamma}\widetilde{q}{(\beta_d)},$ $k=1,\ldots,N$ (see Appendix \ref{ap:feas} for details). A summary of
the complete procedure to optimize an arbitrary entry of the phase code is provided in {\bf Algorithm \ref{alg_CPM}}.
 \begin{algorithm}
\caption{Continuous Phase Code Entry Optimization}
\label{alg_CPM}
\textbf{Input:} Initial code vector $\bx^{(n)}$, code entry $d$, $\theta$, and accuracy $\epsilon_1$;\\
\textbf{Output:} Optimal solution $x_d^{\star}$;
\begin{enumerate}
    \item Compute $u = f_{\theta}\left(\bx^{(n)}\right)$ as well as $\beta_d=\tan\left(\frac{\arg(x_d^{(n)})}{2}\right)$, and set $w = 0$;
        \item \textbf{do}
\begin{enumerate}
    \item $\bar {\gamma} = \frac{u+w}{2}$;
    \item $ \overline{\cal{A}}_k^{ \bar{\gamma}} = \left\{ \beta_d : \left[ \widetilde{p}_{k}(\beta_d) - \bar{\gamma}\widetilde{q}(\beta_d)\right] > 0 \right\}$;
    \item if ${\bigcup_{k = 1}^{N-1}\overline{\cal{A}}_k^{ \bar{\gamma}}}\neq {\cal{R}}$, $u = \bar{\gamma}$ and pick up a feasible solution  $\beta_d$; \textbf {else} $w = \bar{\gamma}$;
\end{enumerate}
\item \textbf{until} $u - w \leq \epsilon_1$;
    \item  Set $\phi^{\star}_d = 2 \mbox{atan}\left(\beta_d^{\star}\right) $, with $\beta_d^{\star}$ the obtained $\epsilon_1$-optimal solution;
    \item Set $x_d^{\star} = e^{\jmath \phi^{\star}_d} $.
\end{enumerate}
\end{algorithm}
\\
{\bf Remark 1.} To establish the computational complexity of {\bf Algorithm \ref{alg_CPM} } it is necessary to observe that the main actions in its implementation are:
\begin{enumerate}
\item calculation of  ${\cup_{k = 1}^{N-1}\overline{\cal{A}}_k^{ \bar{\gamma}}}$;
\item bisection iterations;
\item evaluation of the optimal phase.
\end{enumerate}
Calculation of $\overline{\cal{A}}_k^{ \bar{\gamma}}$, $k=1,\ldots,N-1$, involves the evaluation of the roots of the fourth order polynomial $\widetilde{p}_{k}(\beta_d) - \bar{\gamma}\widetilde{q}(\beta_d)$ which can be done in closed form via Cardano's procedure \cite{shmakov2011universal}.
Hence, ${\cup_{k = 1}^{N-1}\overline{\cal{A}}_k^{ \bar{\gamma}}}$  can be obtained ordering the resulting real roots (possibly merging some overlapping intervals) with an overall computational complexity ${\cal{O}}(N\log_2(N))$ in the worst case \cite{seidel2005top}.  As to the bisection method, in each iteration the search interval is divided in two parts. As a consequence, the interval size after $\bar{n}$ iterations is $2^{-\bar{n}}(u_{(0)}-w_{(0)})$. It follows that ${\cal{K}}=\lceil\log_2{(\frac{u_{(0)}-w_{(0)}}{\epsilon_1})}\rceil$ iterations are required before the algorithm terminates. Finally, since each step involves the solution of $\widetilde{P}_{d,\beta_d}^{\infty,\theta}$, the overall complexity is ${\cal{O}}({\cal{K}}N\log_2(N))$.
 \subsection{Discrete Phase Code Design}
Let us now consider Problem $P^{M,\theta}_{d,\bx^{(n)}}$ and develop an efficient procedure to find its optimal solution exploiting Discrete Fourier Transform (DFT)\footnote{Note that, performing quantization of a good continuous phase sequence does not guarantee a good discrete phase sequence in general.}. In terms of $\phi_d=\arg(x_d)$, $P^{M,\theta}_{d,\bx^{(n)}}$  can be recast as,
\begin{flalign} \label{eq:pbar}
\widetilde{P}^{M,\theta}_{d,\phi_d}
\begin{dcases}
\min_{\phi_d}&g_{\theta}(\phi_d)\\
s.t. & \phi_d \in \bphi_{M}
\end{dcases}
\end{flalign}
where $\bphi_{M} \triangleq \bigg\{0, \frac{2\pi}{M},\frac{4\pi}{M},\ldots,\frac{2\pi(M-1)}{M}\bigg\}$ and,
\begin{IEEEeqnarray*}{rCl}\label{eq:gtpdpm}
g_{\theta}(\phi_{d}) & = & \displaystyle{\max_{k=1,\ldots, N-1}}\bigg[  \theta { \left|{a_{dk}e^{\jmath \phi_d}+b_{dk}e^{-\jmath \phi_d}+c_{dk}} \right|^2} \nonumber\\ && \qquad +\> ( 1- \theta ) \sum_{l=1}^{N-1}{ \left|{a_{dl} e^{\jmath \phi_d}+b_{dl}e^{-\jmath \phi_d}_d+c_{dl}} \right|^2} \bigg].
\end{IEEEeqnarray*}
Evaluating the squared modulus of the autocorrelation in correspondence of the phase variable $\phi_d$ as
\begin{equation} \label{eq:rAB2r3}
\begin{aligned}
\bigg|\widetilde{r}_{k}(\phi_d)\bigg|^2& =\bigg|{r}_{k}(e^{j\phi_d)}\bigg|^2 \\
& =  \left|{a_{dk} e^{\jmath \phi_d}+b_{dk}e^{-\jmath \phi_d}+c_{dk}} \right|^2,
\end{aligned}
\end{equation}
 the following lemma provides a key result to tackle Problem \eqref{eq:pbar}.
\begin{lemma}\label{lem_dpm} {Let $${{\bnu}_{dk}} = [|\widetilde{r}_{k} (\bar{\phi}_1)|^2, |\widetilde{r}_{k} (\bar{\phi}_2)|^2, \ldots, |\widetilde{r}_{k} (\bar{\phi}_{M})|^2]^T \in {\cal{R}}^{M},$$ with $\bar{\phi}_i=\frac{2 \pi (i-1)}{M}$, $i=1,\ldots,M$, and $\bzeta_{dk} = [a_{dk} , c_{dk}, b_{dk}, \bzero_{1\times (M-3)}]^T \in {\cal{R}}^{M}$. If $M \geq 3$, then}
 \begin{equation} \label{eq:FFTproof}
 {\bnu}_{dk} = |\text{DFT}(\bzeta_{dk})|^2,
 \end{equation}
  where $\text{DFT}(\bzeta_{dk})$ is the $M$-points DFT of the vector $\bzeta_{dk}$ and the square modulus is element wise.
  \end{lemma}
 \begin{proof}
  See Appendix \ref{ap:dpm}.
  \end{proof}
Now, defining the matrix $\bU \in \mathbb{R}^{(N-1)\times M} $ whose $k$th row is
 $$\bu^{k} = \theta { {\bnu}_{dk}^T }+ ( 1- \theta ) \sum_{l=1}^{N-1}{ {\bnu}_{dl}^T }\in \mathbb{R}^M,\,\,\,\,k=1,\ldots,N-1,$$
 the optimal solution to $\widetilde{P}^{M,\theta}_{d,\phi_d}$ is given by
\begin{equation}\label{eq:phioptdis}
\phi_{d}^{\star} = \frac{2\pi (i^{\star}-1)}{M},
\end{equation}
where
\begin{equation}\label{eq:i_star}
i^{\star}=\arg \displaystyle{\min_{i=1,\ldots,M}}\Big\{\max\left(\bu_{i}\right)\Big\},
\end{equation}
and $\bu_{i}\in \mathbb{R}^{(N-1)}$ is the $i$th column of $\bU$. Hence, based on Lemma \ref{lem_dpm} and \eqref{eq:phioptdis}, the optimal phase code entry can be efficiently computed as $x_d^{\star} = e^{\jmath \phi_d^{\star}}$ using DFT. In \textbf{Algorithm \ref{alg_DPM}} the proposed approach is reported.\\
{\bf Remark 2.} According to Lemma \ref{lem_dpm}, the developed approach assumes $M\geq 3$. To design binary phase sequences a slight modification of {\bf Algorithm \ref{alg_DPM}} is required. To this end, observe that when $x_d$ is a real binary  variable
\begin{equation} \label{eq:CPM:rdk2}
\begin{aligned}
{r}_{k}(x_d) = &  x_d (x_{d+k}\mathbf{1}_{A}(d+k)+x_{d-k}\mathbf{1}_{A}(d-k)) \\& + \sum_{i=1 , i \neq d , d - k}^{N-k}{x_i x_{i+k}}, ~~~ k = 1 , \ldots , N-1,
\end{aligned}
\end{equation}
implying that
\begin{equation}
|\widetilde{r}_{k}(\phi_d)|^2= \left|{\bar{a}_{dk} e^{\jmath \phi_d}+\bar{c}_{dk}} \right|^2,
\end{equation}
 with $\phi_d \in \{0, \pi\}$ and $\bar{a}_{dk} = x_{d+k}\mathbf{1}_{A}(d+k) + x_{d-k}\mathbf{1}_{A}(d-k)$, $\bar{c}_{dk}  = \sum_{i=1 , i \neq d ,  d - k}^{N-k}{x_i x_{i+k}}$ real coefficients. As a consequence, it is sufficient to update in Lemma \ref{lem_dpm} the definition of the vector  ${\bzeta}_{dk}$ as $\bar{\bzeta}_{dk} = [\bar{a}_{dk} ,\bar{c}_{dk}]^T \in {\cal{R}}^{2}$.\\
{\bf Remark 3.} {\bf Algorithm 3} needs the evaluation of $(N-1)$ different $M$-points DFTs. Each of them can be efficiently computed  via a Fast Fourier Transform (FFT).  Therefore the computational complexity order  is  ${\cal{O}}(NM\log_2{M})$ \cite{refFFTcomp}.
 \begin{algorithm}
\caption{Discrete Phase Code Entry Optimization}
\label{alg_DPM}
\textbf{Input:} Initial code vector $\bx^{(n)}$, code entry $d$, $\theta$, and $M$;\\
\textbf{Output:} Optimal solution $x_d^{\star}$;
\begin{enumerate}
    \item { If $M \geq 3$ then $\forall k \in \left\{1,\ldots,N-1\right\}$
    \begin{itemize}
    \item{Set $a_{dk} = x_{d+k}^{*}\mathbf{1}_{A}(d+k)$, $b_{dk} = x_{d-k}\mathbf{1}_{A}(d-k)$ and $c_{dk}  = \sum_{i=1 , i\neq d , d - k}^{N-k}{x_i x_{i+k}^*}$;}
    \item{Set $ \bzeta_{dk} = [a_{dk} , c_{dk}, b_{dk}, \bzero_{1\times (M-3)}]^T$;}
    \item{Set ${\bnu}_{dk} = |\text{FFT}(\bzeta_{dk})|^2$;}
    \end{itemize}
    Else, if $M = 2$ then
    \begin{itemize}
    \item{Set $\bar{a}_{dk} = x_{d+k}\mathbf{1}_{A}(d+k)+x_{d-k}\mathbf{1}_{A}(d-k)$ and $\bar{c}_{dk}  = \sum_{i=1 , i \neq d , d - k}^{N-k}{x_i x_{i+k}}$;}
    \item{Set $ \bar{\bzeta}_{dk} = [\bar{a}_{dk} , \bar{c}_{dk}]^T$;}
    \item{Set ${\bnu}_{dk} = |\text{FFT}(\bar{\bzeta}_{dk})|^2$;}
    \end{itemize}
    }
    \item Calculate $\bu^{k} = \theta { {\bnu}_{dk}^T }+ ( 1- \theta ) \sum_{l=1}^{N-1}{ {\bnu}_{dl}^T }\in \mathbb{R}^M,\,\,\,\,  k=1,\ldots,N-1$ and $ \bomega_d = [\max\{\bu_{1}\}, \max\{\bu_{2}\}, \ldots , \max\{\bu_{M}\}]^T $;
    \item Find the index $i^{\star}$ where $\bomega_d $ is minimum;
    \item Set $x_d^{\star} = e^{\jmath \phi^{\star}_d} $ with $\phi_{d}^{\star} = \frac{2 \pi (i^{\star}-1)}{M}$.
\end{enumerate}
\end{algorithm}
\subsection{$l_p$-norm Minimization for Algorithm Initialization}\label{heuristic_approach}
The solution obtained via the designed method  depends evidently on the initial sequence. As a result,
 the development of a heuristic approach that can be used to provide high quality starting points is valuable. To this end, recall that the minimization of the $l_p$-norm of the  autocorrelation vector $[r_1, r_2, \ldots, r_{N-1}]$ allows to trade-off ISL and PSL values of the devised sequence as the value of $p$ increases \cite{7362231,5765722,4383616}. Besides, the PSL coincides with the limit as $p \rightarrow  \infty$ of the autocorrelation vector $l_p$-norm.
 According to the above considerations, a procedure to obtain  phase-only codes with low autocorrelation $l_p$-norm is introduced. In particular, with reference to the PSL metric, a start-stop procedure involving a sequence of $l_p$-norm minimization problems with increasing value of $p$, $p_1 < p_2 < \ldots < p_e$ say,  is employed  similarly to \cite{7362231}. 
   Specifically, the algorithm is initialized with Frank, Golomb or a random sequence and the $l_p$-norm minimization starts with $p=2$, i.e., $p_1=2$. Then, $p$ is set to $p_2$ and the algorithm starts with the solution obtained for $p=p_1$, and so on. In general, the $l_p$-norm minimization problem for continuous and discrete cases, can be formulated as
 \begin{flalign}\label{eq:lpOpt}
H^{\infty,p}
\begin{dcases}
\min_{\bx}&\sum_{k=1}^{N-1}{|r_{k}|^p}\\
s.t. & \bx \in \Omega_{\infty}
\end{dcases}
&,
H^{M,p}
\begin{dcases}
\min_{\bx}& \sum_{k=1}^{N-1}{|r_{k}|^p}\\
s.t. & \bx \in \Omega_{M}
\end{dcases}
\end{flalign}
To tackle $H^{\infty,p}$ and $H^{M,p}$ the framework proposed in \cite{razaviyayn2013unified} is exploited, where each variable block corresponds to one code entry and the surrogate function of \cite{7362231} is adopted. Specifically, at step $n+1$ of the iterative procedure, the following optimization problems are considered,
\begin{flalign}\label{eq:lp2}
H^{\infty,p}_{d,\bx^{(n)}}
\begin{dcases}
\min_{x_d}&\sum_{k=1}^{N-1}{\widetilde{\tau}_k |r_{k}|^2 +  \widetilde{\lambda}_k \Re\left\{ r_k^* \frac{r_k^{(n)}}{\left|r_k^{(n)}\right|} \right\} }\\
s.t. & |x_d|  = 1
\end{dcases}
\end{flalign}
\begin{flalign}
H^{M,p}_{d,\bx^{(n)}}
\begin{dcases}
\min_{x_d}& \sum_{k=1}^{N-1}{\widetilde{\tau}_k |r_{k}|^2 + \widetilde{\lambda}_k |r_{k}| }\\
s.t. & x_d \in \{1, e^{\jmath\frac{2\pi}{M}},\ldots,e^{\jmath\frac{2\pi(M-1)}{M}}\}
\end{dcases}
\end{flalign}
where
\begin{equation}
\widetilde{\tau}_k = \frac{t_n^p - \left|r_k^{(n)}\right|^p-p\left|r_k^{(n)}\right|^{p-1}\left(t_n - \left|r_k^{(n)}\right|\right)}{\left(t_n - \left|r_k^{(n)}\right|\right)^2},
\end{equation}
\begin{equation}
\widetilde{\lambda}_k = p\left|r_k^{(n)}\right|^{p-1}-2\widetilde{\tau}_k\left|r_k^{(n)}\right|,
\end{equation}
\begin{equation}
t_n = \left({\sum_{k=1}^{N-1}\left|{r_k^{(n)}}\right|}\right)^{\frac{1}{p}},
\end{equation}
with $\left[r_1^{(n)},r_2^{(n)},\ldots, r_{N-1}^{(n)}\right]^T$ is the optimized autocorrelation vector at step $n$.
Solution techniques solving $H^{\infty,p}_{d,\bx^{(n)}}$ and $H^{M,p}_{d,\bx^{(n)}}$ are now developed.
\subsubsection{Starting Point for the Continuous Phase Code Design}
The first term $\sum_{k=1}^{N-1}\tilde{\tau}_k|r_k|^2$  in the objective of $H^{\infty,p}_{d,\bx^{(n)}}$ can be recast as the ratio of  two quartic polynomials using Lemma \ref{lemm_poly1}. As to the second term, the following lemma is used.
\begin{lemma}\label{lemm_poly2}
{Performing the change of variable $\beta_d \triangleq \tan\left({\frac{\phi_d}{2}}\right)$, $\Re\bigg\{r_k^*\frac{r_k^{(n)}}{|r_k^{(n)}|}\bigg\}$ can be recast as,}
\begin{equation*} \label{eq:rk1}
\Re\left\{r_k^*\frac{r_k^{(n)}}{\left|r_k^{(n)}\right|}\right\} = \frac{\widetilde{\mu}_{dk} \beta_{d}^4 +  \widetilde{\kappa}_{dk} \beta_{d}^3
+ \widetilde{\xi}_{dk} \beta_{d}^2 + \widetilde{\eta}_{dk} \beta_{d} + \widetilde{\rho}_{dk}}{(1+\beta_{d}^2)^2},
\end{equation*}
where $\widetilde{\mu}_{dk},  \widetilde{\kappa}_{dk}  , \widetilde{\xi}_{dk}, \widetilde{\eta}_{dk}, \widetilde{\rho}_{dk}$ are defined in Appendix \ref{ap:rk1}.
\end{lemma}
\begin{proof}
See Appendix \ref{ap:rk1}.
\end{proof}
Now, based on Lemma \ref{lemm_poly1}, Lemma \ref{lemm_poly2}, as well as equations \eqref{eq:pkbeta} and \eqref{eq:qkbeta}, Problem $H^{\infty,p}_{d,\bx^{(n)}}$ is equivalent to,
\begin{flalign}
\widetilde{H}^{\infty,p}_{d,\beta_d}
\begin{dcases}\label{eq:lp3}
\min_{\bx}&\frac{1}{\widetilde{q}(\beta_d)}\sum_{k=1}^{N-1}{\widetilde{\tau}_k p_k(\beta_d) + \widetilde{\lambda}_k h_k(\beta_d)} \\
s.t. & \bx \in \Omega_{\infty}\nonumber
\end{dcases}
\end{flalign}
where
\begin{equation}
h_{k}(\beta_d) =  \widetilde{\mu}_{dk} \beta_{d}^4 + \widetilde{\kappa}_{dk} \beta_{d}^3 + \widetilde{\xi}_{dk} \beta_{d}^2 + \widetilde{\eta}_{dk} \beta_{d} + \widetilde{\rho}_{dk}.
\end{equation}
Finally, Problem $\widetilde{H}^{\infty}_{\beta_d,\widetilde{\gamma}}$ can be efficiently solved using the simplified version of \textbf{Algorithm \ref{alg_CPM}} resulting from the presence of just one fractional quartic function\footnote{Due to the special form of $\widetilde{q}(\beta_d)$ the optimal solution can also obtained finding the real roots of a quartic function (related to the first order derivative of the objective) and evaluating the objective function in these points as well as at $\infty$.}.
\subsubsection{Starting Point for Discrete Phase Code Design}
The discrete phase code design problem can be cast as,
\begin{flalign}
H^{M, p}_{d,\phi_d}
 \begin{dcases}
\min_{x_d}& \sum_{k=1}^{N-1}{\widetilde{\tau}_k |{\widetilde{r}_{k}(\phi_d)} |^2 + \widetilde{\lambda}_k |{\widetilde{r}_{l}(\phi_d)} | }\nonumber \\
s.t. &  \phi_d \in \bphi_{M} \nonumber
\end{dcases}
\end{flalign}
 Using Lemma \ref{lem_dpm} and considering the definition of ${\bnu}_{dk}$ in \eqref{eq:FFTproof}, the optimal $x_d^\star$ can be efficiently obtained as
\begin{equation}
x_d^\star=e^{j\frac{(i^{\star}-1)}{M}},
\end{equation}
with
\begin{equation}
i^{\star}=\arg \displaystyle{\min_{i=1,\ldots,M}}\Big\{\by\Big\},
\end{equation}
and
\begin{equation}
\by =  \displaystyle{\sum_{k=1}^{N-1}}\left(\widetilde{\tau}_{k} { {\bnu}_{dk}} + \widetilde{\lambda}_k\sqrt{{\bnu}_{dk}}\right).
\end{equation}

\section{Performance Analysis}\label{perf_analysis}
This section is devoted to the performance analysis of the proposed algorithms for both continuous and discrete phase code design\footnote{In the following, Continuous Phase Method (CPM) and Discrete Phase Method (DPM) refer to {\bf Algorithm 1}
used to solve $P^{\infty,\theta}$ (based on {\bf Algorithm 2}) and $P^{M,\theta}$ (based on {\bf Algorithm 3}), respectively. In both cases, the heuristic initialization given in Subsection \ref{heuristic_approach} is exploited unless otherwise stated.} exploiting PSL and ISL as performance measures. Moreover, a comparison with the state of art algorithms available in the open literature, i.e., CAN \cite{5072243}\footnote{The
Matlab code for CAN is downloaded from the website (http://www.sal.ufl.edu/book/).}, ITROX \cite{6142119}, MWISL-Diag, and MM-PSL \cite{7362231}\footnote{MWISL-Diag and MM-PSL have been simulated according to Algorithm 2 and Algorithm 4 in \cite{7362231}, respectively.} is conducted. All the considered procedures are initialized using the same set of starting codes
 composed by Golomb and Frank\footnote{Note that Frank sequences are defined for
lengths that are perfect squares.} codes as well as $5$ random sequences (if not differently specified). Hence, the best obtained objective value is reported and the resulting sequence picked up. Finally, the stopping criteria $|\text{obj}(\bx^{(n)})-\text{obj}(\bx^{(n-1)})| \leq 10^{-5}$ is used to terminate  all the algorithms.
\subsection{PSL Minimization}
In this subsection, the ability of the proposed algorithms to synthesize low PSL sequences is assessed. To this end, the Pareto weight is fixed to $\theta=1$ and the sequence of $p$-values for the selection of the initial starting point is $2,2^{2},2^{3},\ldots,2^{13}$, i.e., $p_i=2^i$, $i=1,\ldots,13$  (see Subsection \ref{heuristic_approach}). In Fig. \ref{Fig:psl1_a}.(a), the PSL versus the alphabet size $M$ of the devised sequences is displayed assuming $N=400$; precisely, $M = \{2, 4, 8, 16, 32, 64, 128, 256, 512\}$ is considered. For comparison purposes, the PSL values of CAN-based techniques\footnote{The procedure based on CAN algorithm providing discrete phase sequences is referred to as CAN(D).}, MM-PSL procedures, and CPM algorithm are reported. The results highlight that CPM outperforms CAN, CAN(D), MM-PSL, and DPM. Additionally, DPM ensures better performance than CAN(D) and  achieves lower PSL values than CAN and MM-PSL as the alphabet set is dense enough. Finally, CPM provides a lower bound to DPM which is tighter and tighter as the alphabet size increases, i.e., DPM converges to CPM when  $M\rightarrow \infty$.
 \begin{figure}%
\centering
\subfigure[PSL (dB) versus constellation size.]{\includegraphics[width = 85mm]{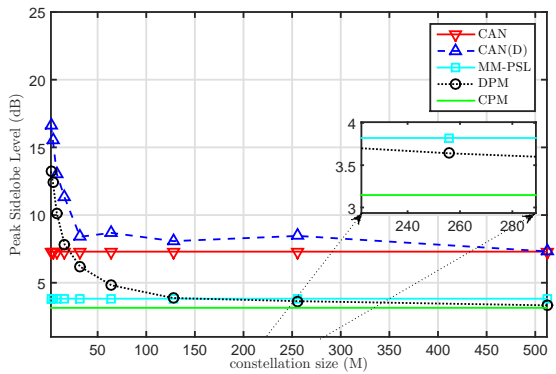}}\qquad
\subfigure[PSL (dB) versus iteration.]{\includegraphics[width = 85mm]{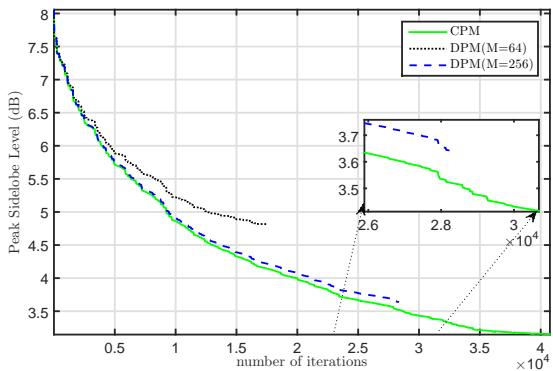}}\    %
\caption{The effects of constellation size on PSL values.}
\label{Fig:psl1_a}
\end{figure}\\
In Fig. \ref{Fig:psl1_a}.(b) the convergence behavior  after heuristic initialization of CPM and DPM, for $M = 64$ and $M = 256$,  is analyzed. As expected,  the lower the alphabet size the faster the convergence but the worst the obtained PSL. Indeed, increasing $M$ (CPM is tantamount to considering $M=\infty$) the feasible set of DPM becomes larger and larger enabling better and better PSLs; nevertheless, more and more iterations are required to explore the enlarged domain. Finally, the convergence curves illustrate the monotonic decreasing behavior of the objective function.\\
The second experiment provides the PSL for sequence lengths $[5^2,7^2,10^2,12^2,15^2,18^2,20^2,25^2,30^2,32^2]$. The results of CPM, DPM, CAN, and MM-PSL are shown in Fig. \ref{Fig:psl2}.
The plot reveals that CPM exhibits a performance level better than the counterparts for all the considered lengths. Meanwhile, DPM with alphabet size $M=256$ obtains usually better PSL as compared with MM-PSL.
\begin{figure}%
\centering
{\includegraphics[width = 85mm]{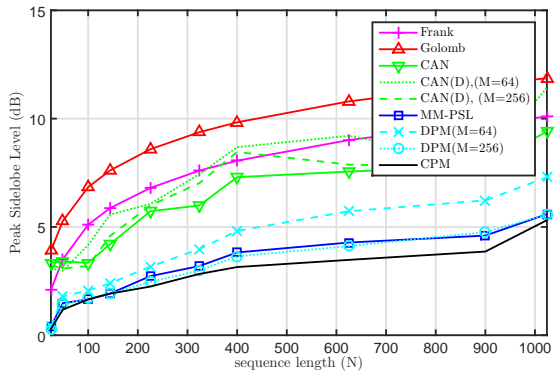}}\qquad
\caption{PSL (dB) versus sequence length.}
\label{Fig:psl2}
\end{figure}
\begin{figure}%
\centering
\subfigure[$M = 2$.]{\includegraphics[width = 85mm]{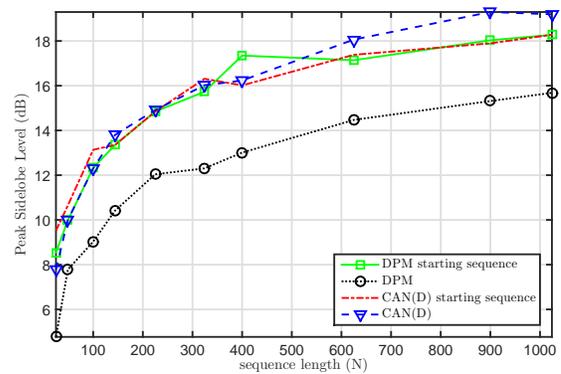}}\qquad
\subfigure[$M = 8$.]{\includegraphics[width = 85mm]{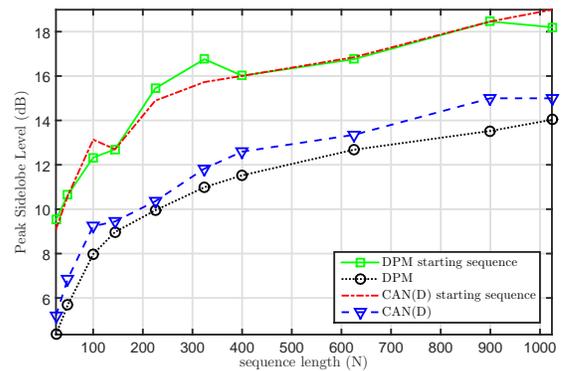}}\qquad
\subfigure[$M = 16$.]{\includegraphics[width = 85mm]{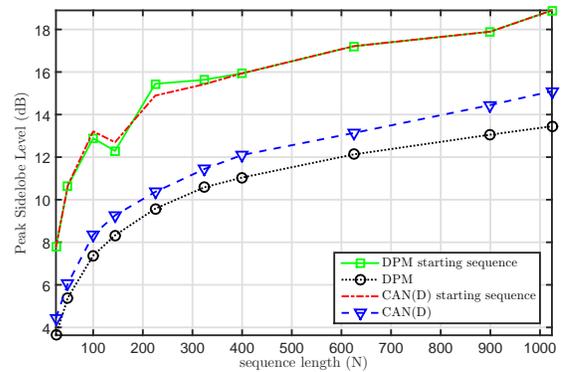}}\    %
\caption{A comparison between CAN(D) and DPM starting from the same random sequence.}
\label{Fig:psl3}
\end{figure}

Next, the capability of DPM to design discrete phase sequences is assessed. In Fig. \ref{Fig:psl3}, the PSL versus the code length of  CAN(D) and DPM are reported for alphabet sizes $M = 2$, $M = 8$, and $M = 16$.
For each $M$, the same set of $5$ random sequences (drawn by a uniform distribution over the set of the feasible sequences) is used to initialize both the algorithms. The results show that DPM outperforms CAN(D) significantly with maximum gains of 3.98 dB, 1.47 dB, and 1.65 dB for $M=2$, $M=8$, and $M=16$, respectively.\\
The case of binary phase sequence design is further investigated considering ITROX \cite{6142119} (the only optimization-based algorithm currently available in the open literature that provides good binary phase sequences) as benchmark.
Precisely, in Fig. \ref{Fig:psl4}.(a) the PSL versus $N$ is displayed for both ITROX and DPM where the same set of $5$ binary random codes is used for initialization. To highlight the quality of DPM algorithm also the PSL of MPS sequences, obtained via exhaustive search up to the length of $105$, is shown in this figure. The plot clearly illustrates the effectiveness of our approach. Indeed, DPM significantly outperforms ITROX and provides a PSL quite close to the global optimum of MPS sequences but with  a much lower computational complexity and without restrictions to the maximum code length. This last feature is particularly appealing since the higher $N$ the higher the pulse compression.
Interestingly, DPM provides in some circumstances the global optimal solution (see in Fig. \ref{Fig:psl4} the points where DPM and MPS coincide). In this respect, in  Fig. \ref{Fig:psl4}.(b)  the autocorrelation function of a binary random sequence which leads to the Barker sequence of length $11$ though DPM  is depicted (notice that 15\% of the trials leads to this PSL value via DPM). As another example, in Fig. \ref{Fig:psl4}.(c) the autocorrelation function devised via DPM for sequence length $126$ is displayed. Remarkably, the PSL is equal to 8 (4\% of the trials leads to this PSL value via DPM and only 10\% is higher than or equal to 11) whereas the best PSL that \cite{nasrabadi2010survey} obtains  for the same sequence length using genetic algorithms is $11$ which further confirms the effectiveness of the new framework.

\begin{figure}%
\centering
\subfigure[A comparison among the PSL values of binary sequences obtained through DPM,  ITROX-AP, and exhaustive search (MPS). For each length, DPM and ITROX-AP have been run 5 times and from the resulting 5 PSL values the best one is chosen.]{\includegraphics[width = 85mm]{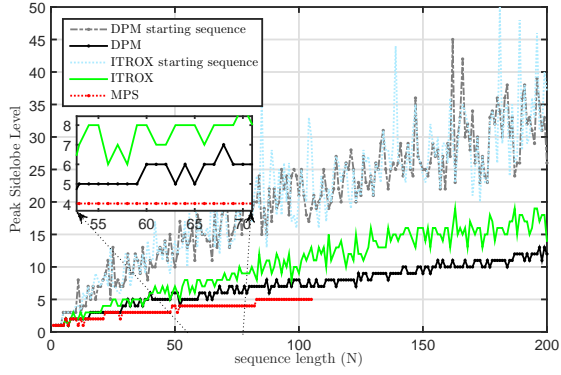}}\    \subfigure[Binary code design of length 11.]{\includegraphics[width = 85mm]{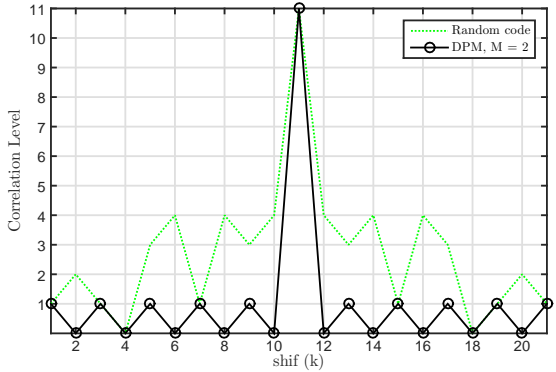}}\    %
\subfigure[Binary code design of length 126.]{\includegraphics[width = 85mm]{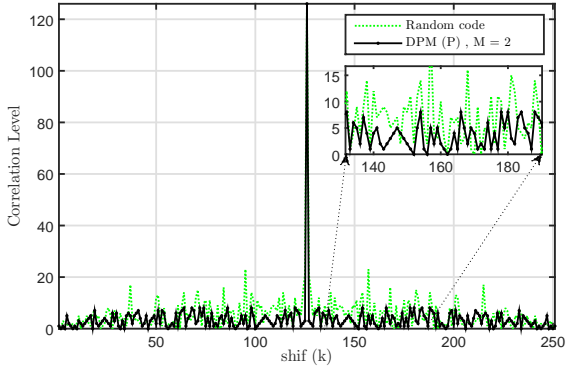}}\    %
\caption{Performance of DPM algorithm in different conditions.}
\label{Fig:psl4}
\end{figure}

\subsection{ISL Minimization}
The performance assessment of CPM and DPM for ISL minimization is now considered. In this case, $\theta = 0$ and the initialization procedure in Subsection \ref{heuristic_approach} is not performed. In Fig. \ref{Fig:isl1}.(a), the ISL versus $M$ is displayed for DPM, CAN(D), CAN, WISL-Diag\footnote{The weights of MWISL-Diag are set to one so as to account for ISL minimization.}, and CPM assuming $N = 400$. Interestingly, the continuous phase design strategies exhibit almost the same performance. As to the discrete phase code design, DPM outperforms the counterpart, i.e., CAN(D)\footnote{Notice that, the discrete phase counterpart to MWIS-Diag is not provided in \cite{7362231}.}. Specifically, CAN(D) requires a larger constellation size than DPM to achieve the same ISL value. In Fig. \ref{Fig:isl1}.(b), the convergence behavior of CPM and DPM, for $M=64$ and $M=256$, is plotted and similar considerations to those for Fig. \ref{Fig:psl1_a}.(b) hold true.
\begin{figure}%
\centering
\subfigure[ISL versus constellation size.]{\includegraphics[width = 85mm]{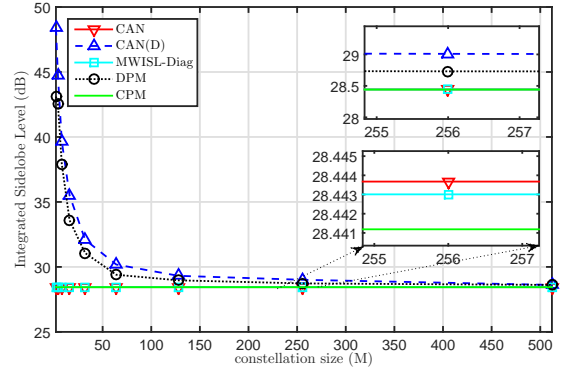}}\qquad
\subfigure[Convergence curve for ISL minimization with $N = 400$.]{\includegraphics[width = 85mm]{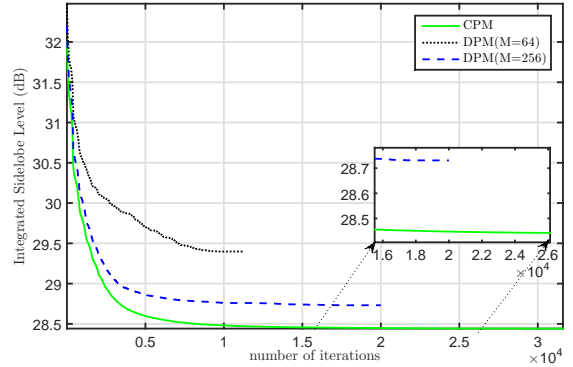}}\    %
\caption{The effect of constellation size on ISL values.}
\label{Fig:isl1}
\end{figure}
To further corroborate the effectiveness of our strategies, in Fig. \ref{Fig:isl2} the ISL versus $N$ is illustrated. In line with the previous results, the continuous phase design strategies are almost equivalent for all $N$ as well as  DPM outperforms CAN(D).\\
\begin{figure}%
\centering
{\includegraphics[width = 85mm]{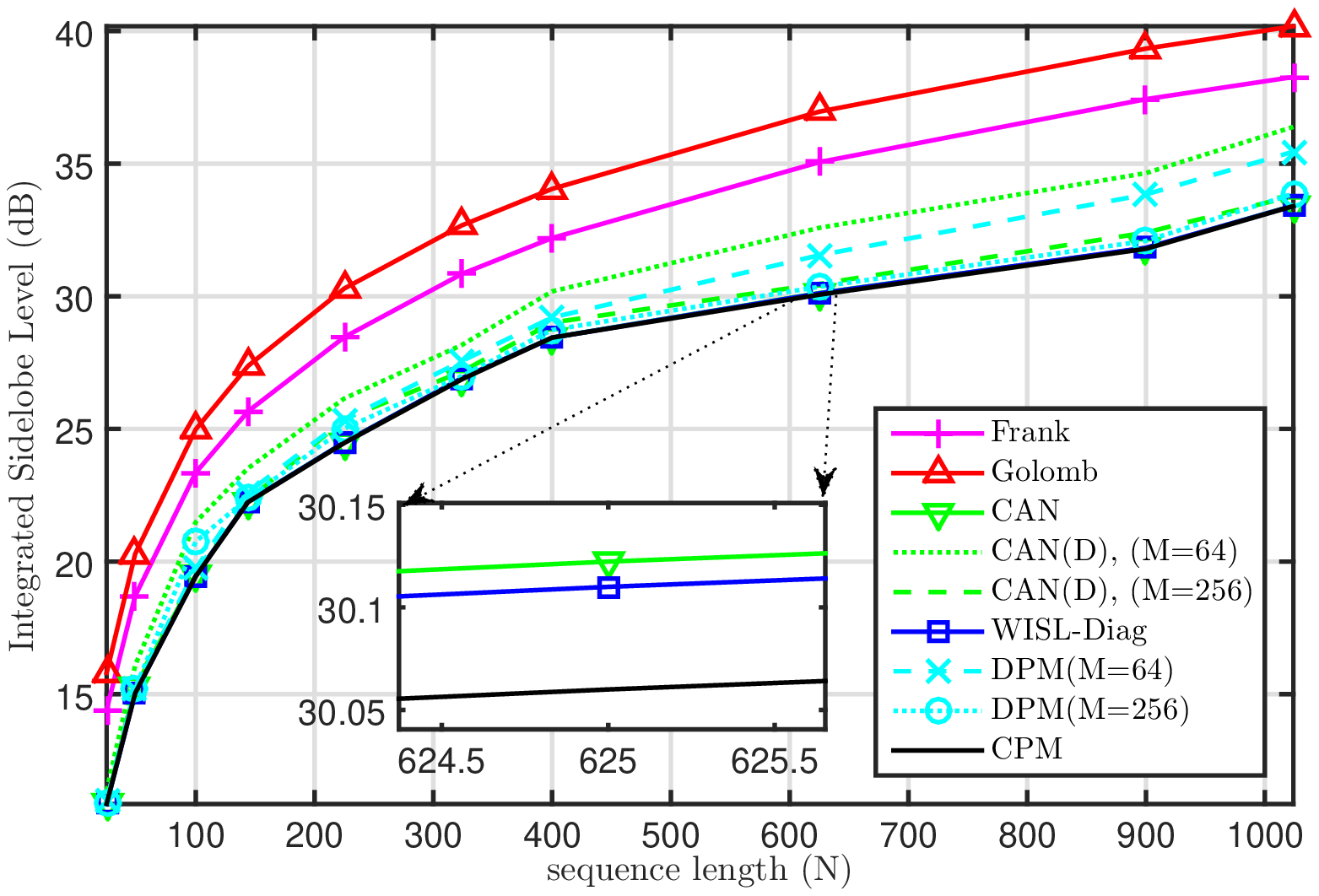}}\qquad
\caption{ISL (dB) versus sequence length.}
\label{Fig:isl2}
\end{figure}
Finally, Figs. \ref{Fig:isl3}.(a), \ref{Fig:isl3}.(b), and  \ref{Fig:isl3}.(c) display the ISL versus $N$ for $M=2$, $M=8$, and $M=16$, respectively. $5$ random sequences (drawn by a uniform distribution over the set of the feasible sequences) are considered as starting points and the best resulting code is picked up. The results clearly show the significant performance gain granted by DPM. Specifically, for binary phase codes the maximum ISL gain of DPM with respect to CAN(D) is 5.99 dB, whereas for $M=8$ and $M=16$ the gains are 1.73 dB and 1.4 dB, respectively.
\begin{figure}%
\centering
\subfigure[ISL versus $N$ for binary phase sequence designed through CAN(D) and DPM.]{\includegraphics[width = 85mm]{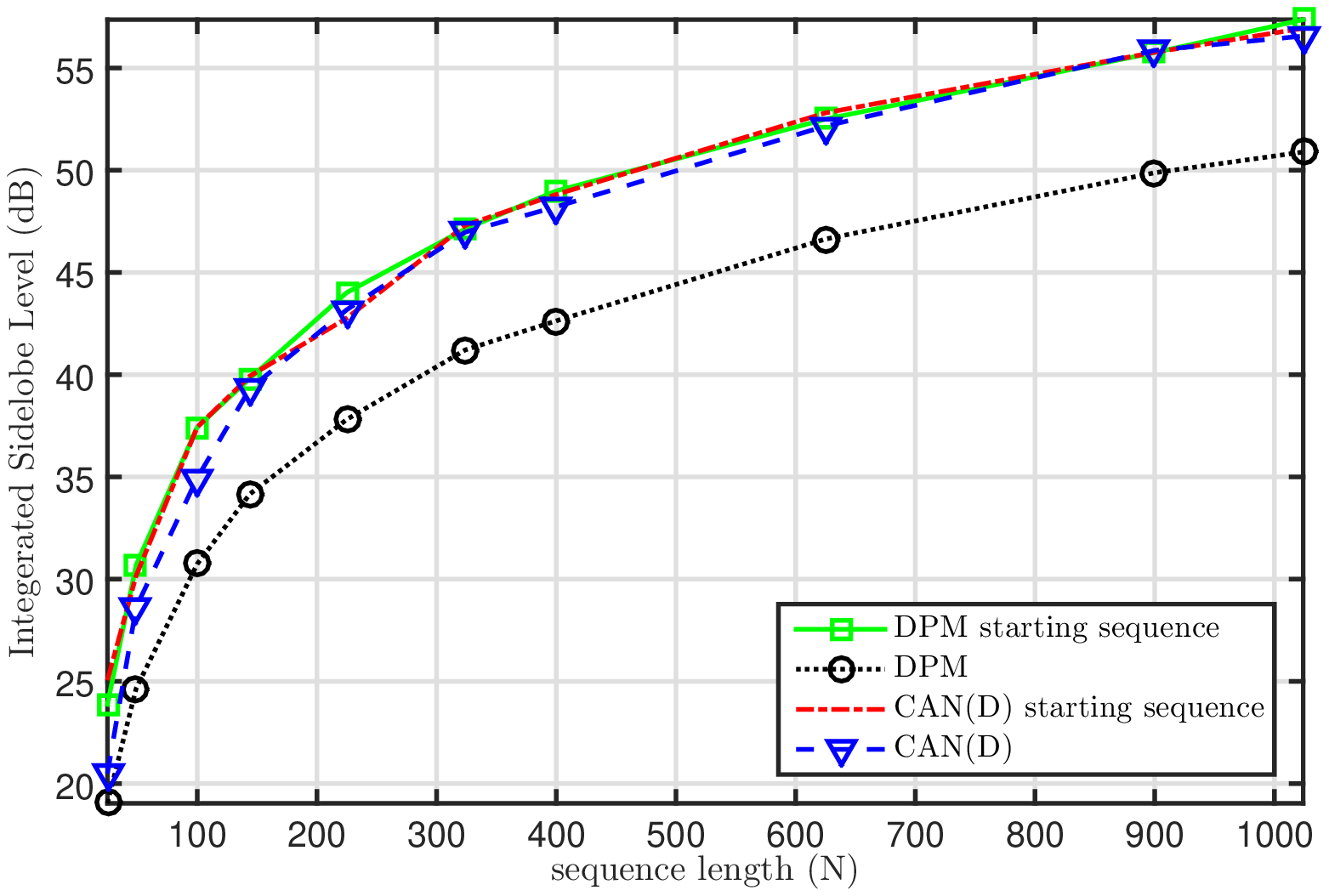}}\qquad
\subfigure[ISL versus $N$ for discrete phase sequences with $M=8$ designed through CAN(D) and DPM.]{\includegraphics[width = 85mm]{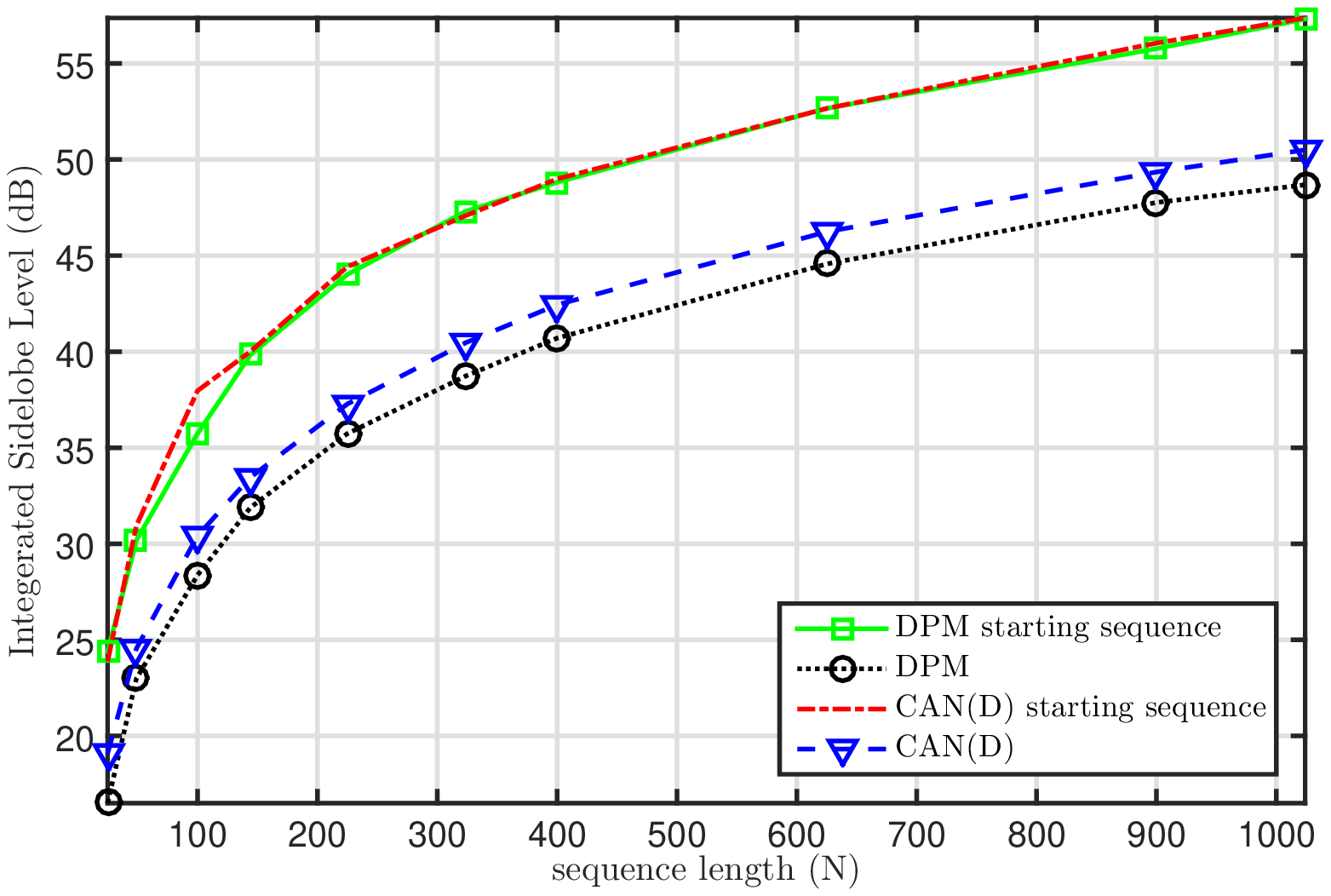}}\qquad
\subfigure[ISL versus $N$ for discrete phase sequences with $M=16$ designed through CAN(D) and DPM.]{\includegraphics[width = 85mm]{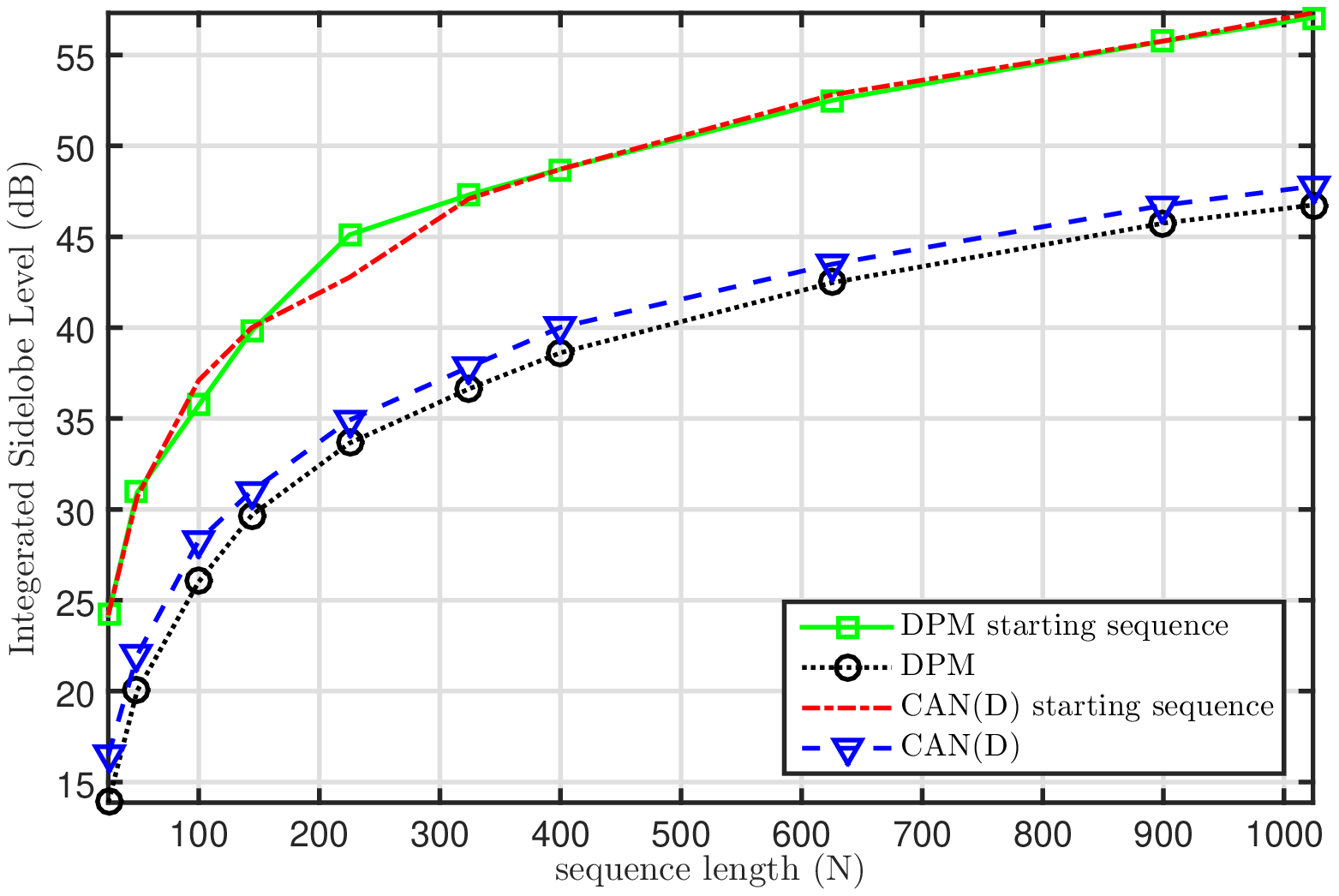}}\
\caption{ISL (dB) versus sequence length.}
\label{Fig:isl3}
\end{figure}

\subsection{Pareto-Optimal Solution}
 In this subsection, the impact of the parameter $\theta$ on the designed codes is illustrated. Precisely, in Fig. \ref{Fig:pareto1} the Pareto curves obtained via CPM and DPM, for $M=64$ and $M=256$, are shown assuming $N = 400$ and ${\theta}\in\{\theta_1,\theta_2,\ldots,\theta_{{6}}\}\subseteq [0,1]$, with $\theta_i=1-(i-1)/5$, $i=1,\ldots6$. The starting sequence used at
 ${\theta}=\theta_{i}$ is the code optimized at ${\theta}=\theta_{i-1}$; also,  at ${\theta}=\theta_{{1}}$ the heuristic approach of Subsection \ref{heuristic_approach} is used. As expected, $\theta$ trades-off ISL and PSL values. Specifically, the higher $\theta$ the better the PSL and the worst the ISL, that is a classical feature of bi-objective Pareto curves.  Otherwise stated, any solution is a Pareto optimal point.
\begin{figure}%
\centering
{\includegraphics[width = 85mm]{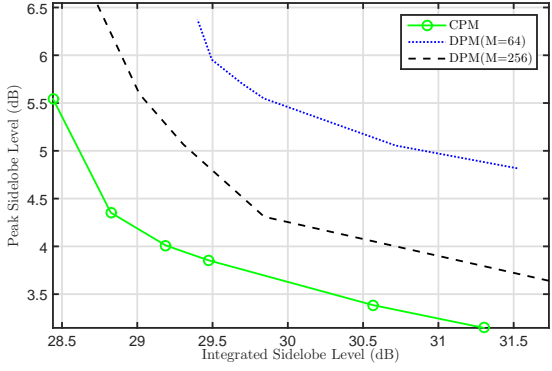}}
\caption{Pareto-optimal curves for continuous and discrete phases ($M=64$ and $M=256$) with $N = 400$.}
\label{Fig:pareto1}
\end{figure}

\section{Conclusion}\label{conclusions}
The synthesis of phase sequences exhibiting good aperiodic correlation features has been addressed. Specifically, PSL and ISL have been adopted as performance metrics and the design problem has been formulated as a bi-objective optimization where either a continuous or a discrete phase constraint is imposed at the design stage. The non-convex and, in general, NP-hard problems resulting from scalarization are handled via a novel iterative procedure based on the CD method.
Each iteration of the devised algorithm requires the solution of a non-convex min-max problem involving quartic functions. With reference to the continuous phase codes design, a new polynomial-time bisection method aimed at solving globally the aforementioned problem is developed. As to the discrete phase case, which includes the challenging and practically valuable binary codes synthesis, an FFT-based procedure is devised. Finally, some heuristic methods based $l_p$-norm minimization have been  introduced to suitably initialize the new converge-ensured algorithms.

At the analysis stage, some interesting case studies have been provided to illustrate the effectiveness of the developed CD design approach. The results highlight the ability of the new procedures to design constant modulus sequences with enhanced
autocorrelation properties. Precisely, the synthesized sequences grant better PSL and ISL than some counterparts
available in the open literature. Besides, these gains are higher and higher as the constellation size reduces.

As future research tracks, it might be interesting to account for  the behavior in the Doppler domain of the synthesized code, i.e., considering the design of codes with a proper ambiguity function, as well as consider a Peak to Average Ratio (PAR) constraint.

\begin{appendices}
\section{Proof of Lemma \ref{lemm_poly1}}\label{ap:rk2}
Let,
\begin{align}
|r_{k}(e^{\jmath\phi_d})|^2 &=  \left|a_{dk}e^{\jmath \phi_d} + {b_{dk} e^{-\jmath \phi_d}+c_{dk}} \right|^2 \nonumber \\
&=  \bigg|\left(a_{dkr}e^{\jmath\phi_d} + b_{dkr}e^{-\jmath\phi_d}+c_{dkr}\right)\nonumber \\
& ~~~~~~ + \jmath\left(a_{dki}e^{\jmath\phi_d}+b_{dki}e^{-\jmath\phi_d}+c_{dki}\right)\bigg|^2 \nonumber \\
& = \bigg( (a_{dkr}+b_{dkr}) \cos(\phi_d) \nonumber \\
& ~~~~~~~~ +(b_{dki}-a_{dki}) \sin{(\phi_d)} +c_{dkr} \bigg)^2 \nonumber \\
& ~~~ + \bigg((a_{dki}+b_{dki})\cos(\phi_d)  \nonumber \\ & ~~~~~~~~ +(a_{dkr}-b_{dkr}) \sin{(\phi_d)} +c_{dki}\bigg)^2 \nonumber \\
& = A_{dk} + B_{dk}
\end{align}
where $a_{dkr} = \Re({a_{dk}})$, $b_{dkr} = \Re({b_{dk}})$, $c_{dkr} = \Re({c_{dk}})$, $a_{dki} = \Im({a_{dk}})$, $b_{dki} = \Im({b_{dk}})$ and $c_{dki} = \Im({c_{dk}})$. Also,
 \begin{equation}
\begin{aligned}\label{eq:Adk1}
A_{dk} = &  \bigg( (a_{dkr}+b_{dkr}) \cos(\phi_d) \\
& +(b_{dki}-a_{dki}) \sin{(\phi_d)} +c_{dkr} \bigg)^2
 \end{aligned}
 \end{equation}
 \begin{equation}\label{eq:Bdk1}
 \begin{aligned}
  B_{dk}  = & \bigg((a_{dki}+b_{dki})\cos(\phi_d) \\
   & +(a_{dkr}-b_{dkr}) \sin{(\phi_d)} +c_{dki}\bigg)^2
\end{aligned}
\end{equation}
Expanding  \eqref{eq:Adk1}
\begin{equation}\label{eq_44}
\begin{aligned}
A_{dk} = &  (a_{dkr}+b_{dkr})^2\cos^2(\phi_d) \\
 & +2(a_{dkr}+b_{dkr})c_{dkr}\cos(\phi_d) \\
 & + c_{dkr}^2+(b_{dki}-a_{dki})^2\sin^2(\phi_d) \\
 & + 2(a_{dkr}+b_{dkr})(b_{dki}-a_{dki})\sin(\phi_d)\cos(\phi_d)\\
 &+2c_{dkr}(b_{dki}-a_{dki})\sin(\phi_d)
\end{aligned}
\end{equation}
Hence, according to the trigonometric relationships \cite{abramowitz1964handbook,ben2001lectures},
\begin{eqnarray}
\sin{\phi_d} && =
\frac{2 \tan\left({\frac{\phi_d}{2}}\right)}{1 + \tan^2\left({\frac{\phi_d}{2}}\right)}
\end{eqnarray}
and
\begin{eqnarray}
\cos{\phi_d} && =
 \frac{1 - 2 \tan^2\left({\frac{\phi_d}{2}}\right)}{1 + \tan^2\left({\frac{\phi_d}{2}}\right)}
\end{eqnarray}
(\ref{eq_44}) can be recast as
\begin{equation}
\begin{aligned}
 A_{dk}  = &  \frac{1}{(1+\beta_{d}^2)^2}\{ (a_{dkr}+b_{dkr})^2(1-\beta_{d}^2)^2 \\
 & +2(a_{dkr}+b_{dkr})c_{dkr}(1-\beta_{d}^4) \\
 & + c_{dkr}^2(1+\beta_{d}^2)^2+4\beta_{d}^2(b_{dki}-a_{dki})^2 \\
 & + 4\beta_{d}(1-\beta_{d}^2)(a_{dkr}+b_{dkr})(b_{dki}-a_{dki})\\
 &+4\beta_{d} c_{dkr}(b_{dki}-a_{dki})\\
 &+4\beta_{d}^3 c_{dkr}(b_{dki}-a_{dki}) \}
  \end{aligned}
\end{equation}
with  $\beta_{d} = \tan\left({\frac {\phi_d}{2}}\right)$.
Besides, using standard algebra it is not difficult to show that,
\begin{equation}\label{Adk}
A_{dk} = \frac{\mu'_{dk} \beta_{d}^4  + \kappa'_{dk} \beta_{d}^3 + \xi'_{dk} \beta_{d}^2 + \eta'_{dk} \beta_{d} + \rho'_{dk}}{(1+\beta_{d}^2)^2} \end{equation}
with
\begin{equation*}
\begin{aligned}
\mu'_{dk} = & (a_{dkr}+b_{dkr})^2 - 2c_{dkr}(a_{dkr}+b_{dkr})+c_{dkr}^2\\
\kappa'_{dk} = & -4(a_{dkr}+b_{dkr})(b_{dki}-a_{dki})+4c_{dkr}(b_{dki}-a_{dki})\\
\xi'_{dk} = & -2(a_{dkr}+b_{dkr})^2 + 2c_{dkr}^2+4(b_{dki}-a_{dki})^2\\
\eta'_{dk} = & 4(a_{dkr}+b_{dkr})(b_{dki}-a_{dki})+4c_{dkr}(b_{dki}-a_{dki})\\
\rho'_{dk} = & (a_{dkr}+b_{dkr})^2+2(a_{dkr}+b_{dkr})c_{dkr}+c_{dkr}^2
\end{aligned}
\end{equation*}
A similar procedure on $B_{dk}$ yields,
\begin{equation}\label{Bdk}
B_{dk} = \frac{ \mu''_{dk} \beta_{d}^4  + \kappa''_{dk} \beta_{d}^3 + \xi''_{dk} \beta_{d}^2 + \eta''_{dk} \beta_{d} + \rho''_{dk}}{(1+\beta_{d}^2)^2}
\end{equation}
where
\begin{equation*}
\begin{aligned}
\mu''_{dk} = & (a_{dki}+b_{dki})^2 - 2c_{dki}(a_{dki}+b_{dki})+c_{dki}^2 \\
\kappa''_{dk} = & -4(a_{dki}+b_{dki})(a_{dkr}-b_{dkr})+4c_{dki}(a_{dkr}-b_{dkr})\\
\xi''_{dk} = & -2(a_{dki}+b_{dki})^2 + 2c_{dki}^2+4(a_{dkr}-b_{dkr})^2\\
\eta''_{dk} = & 4(a_{dki}+b_{dki})(a_{dkr}-b_{dkr})+4c_{dki}(a_{dkr}-b_{dkr})\\
\rho''_{dk} = & (a_{dki}+b_{dki})^2+2(a_{dki}+b_{dki})c_{dki}+c_{dki}^2
\end{aligned}
\end{equation*}
Finally,
\begin{equation}
|\widetilde{r}_{k}(\beta_d)|^2 = \frac{\mu_{dk} \beta_{d}^4 + \kappa_{dk} \beta_{d}^3 + \xi_{dk} \beta_{d}^2 + \eta_{dk} \beta_{d} + \rho_{dk}}{(1+\beta_{d}^2)^2}
\end{equation}
where $\mu_{dk}=\mu'_{dk}+\mu''_{dk}$, $\kappa_{dk}=\kappa'_{dk}+\kappa''_{dk}$, $\xi_{dk} = \xi'_{dk}+\xi''_{dk}$, $\eta_{dk}=\eta'_{dk}+\eta''_{dk}$ and $\rho_{dk} = \rho'_{dk}+\rho''_{dk}$.
\section{Derivation of the feasibility set}\label{ap:feas}
Let
\begin{equation}
 \bar{p}(x) =  a x^4 + b x^3 + c x^2 + d x + e,
 \end{equation}
with $x, a, b, c, d, e \in \mathbb{R}$ as well as $[a, b, c, d, e]^T\neq {\bf 0}$,  and let $ \bar{p}^{'}(x)$, $ \bar{p}^{''}(x)$, $ \bar{p}^{'''}(x)$,  $ \bar{p}^{(4)}(x)$ be the first order, second order, third order, and fourth order derivatives of  $\bar{p}(x)$, respectively. Moreover, denote by $L\leq 4$ the number of distinct real roots of $\bar{p}(x)$ and let $x_i$, $i=1,\ldots,L$, be the ordered real roots.  Since $\bar{p}^{(4)}(x)$ is a continuous function, the following steps allows to construct the set\footnote{Notice that, $x_{0} = -\infty$ and $x_{L+1} = + \infty$.} $\overline{\cal{A}}=\{x\,:\, \bar{p}(x)> 0\}$:
\begin{itemize}
\item[{[1]}]{Let $\overline{\cal{A}}_k =  \varnothing$};
\item[{[2]}]{Find the real roots of $\bar{p}(x)$:}
\begin{itemize}
\item If $L=0$, then if the constant term $e\leq 0$ exit. Conversely, update $\overline{\cal{A}}_k=\mathbb{R}$ and exit;
\item If $L\geq 1$, sort the real roots, set $i=1$, and perform the remaining steps;
 \end{itemize}
\item[{[3]}]{If $ \bar{p}^{'}(x_i) > 0$, then $\bar{p}(x)>0$ on the interval ${(x_i,x_{i+1})}$: update $\overline{\cal{A}}_k = \overline{\cal{A}}_k \cup (x_{i}, x_{i+1})$; if $i < L $ set $i=i+1$ and repeat step [3], otherwise exit;}
\item[{[4]}]{If $ \bar{p}^{'}(x_i) < 0$, then $\bar{p}(x_i)>0$ on the interval ${(x_{i-1},x_{i})}$:  update $\overline{\cal{A}}_k = \overline{\cal{A}}_k \cup (x_{i-1}, x_{i})$; if $i < L$ set $i=i+1$ and repeat step [3], otherwise exit;}
\item[{[5]}]{If $ \bar{p}^{'}(x_i) = 0$, then $x_i$ is a stationary point:}
\begin{itemize}
\item{If $ \bar{p}^{''}(x_i) > 0 $, then  $x_i$ is a local minimum: update $\overline{\cal{A}}_k = \overline{\cal{A}}_k \cup (x_{i-1},x_{i}) \cup (x_{i},x_{i+1})$; if $i < L$ set $i=i+1$ and repeat step [3], otherwise exit;}
\item{If $ \bar{p}^{''}(x_i) < 0 $, then $x_i$ is a local maximum: if $i < L $ set $i=i+1$ and repeat step [3], otherwise exit;}
\item{If $ \bar{p}^{''}(x_i) = 0 $:}
\begin{itemize}
\item{If $ \bar{p}^{'''}(x_i) \neq 0$, then $x_i$  is an inflection point. If $ \bar{p}^{'''}(x_i) > 0$:  update $\overline{\cal{A}}_k = \overline{\cal{A}}_k \cup (x_{i}, x_{i+1})$; if $i < L$ set $i=i+1$ and repeat step [3], otherwise exit. Conversely if $ \bar{p}^{'''}(x_i) < 0$}: update $\overline{\cal{A}}_k = \overline{\cal{A}}_k \cup (x_{i-1}, x_{i})$; if $i < L$ set $i=i+1$ and repeat step [3], otherwise exit;
\item{If $ \bar{p}^{'''}(x_i) = 0 $,}
 \begin{itemize}
 \item{If $ \bar{p}^{(4)}(x_i) > 0$, then  $x_i$ is a local minimum: update $\overline{\cal{A}}_k = \overline{\cal{A}}_k \cup (x_{i-1},x_{i}) \cup (x_{i},x_{i+1})$; if $i<L$ set $i=i+1$ and repeat step [3], otherwise exit. Conversely, if $ \bar{p}^{(4)}(x_i) < 0 $ then $x_i$ is a local maximum: if $i < L $ set $i=i+1$ and repeat step [3], otherwise exit.}
     \end{itemize}
    \end{itemize}
    \end{itemize}
\end{itemize}
In order to calculate the union of the different sets, the fast and simple ``union-find'' algorithm \cite{seidel2005top} is employed. Precisely, let $(l_1,u_2), (l_2,u_2), \ldots, (l_{M},u_{M})$ be $M$ different intervals where $l_i$ is the lower bound of each set  and $u_i$ is the upper bound. Let $\bt \in {\cal{R}}^{M}$ be the vector containing the sorted $l_i$ and $u_i$, $i = 1, \ldots, M$ in increasing order (if $l_i=u_k$ then $u_k$ is located first than $l_i$). Now, define a counter $Count$  initialized as $Count=1$; then check if the second entry of $\bt$ is a left extreme or a right extreme of one of $M$ intervals. If it is a left extreme, $Count=Count+1$ otherwise $Count=Count-1$. Now, if $Count=0$ an interval disjoint from the remaining part of the set is obtained and the process continues for the construction of the remaining part of the union set starting from the successive entry of $\bt$.


\section{Proof of Lemma  \ref{lem_dpm}}\label{ap:dpm}
The $M$-point DFT of $\bzeta_{dk}$ is,
\begin{equation*}
  {\cal{F}}_{M}(\bzeta_{dk})
                        = \left[
                          \begin{array}{c}
                            a_{dk}+c_{dk}+b_{dk} \\
                            a_{dk}+c_{dk}e^{-\jmath \frac{2 \pi}{M} }+b_{dk}e^{-\jmath \frac{4 \pi}{M} } \\
                            \vdots \\
                            a_{dk}+c_{dk}e^{-\jmath \frac{2 \pi(M-1)}{M} }+b_{dk}e^{-\jmath \frac{4 \pi(M-1)}{M} }\\
                          \end{array}
                        \right]
\end{equation*}
Next, observe that
\begin{equation} \label{eq:rkproof2}
{\widetilde{r}_{k}}(\bar{\phi_m}) e^{-\jmath \bar{\phi_m}} = a_{dk} +c_{dk}e^{-\jmath \bar{\phi_m}}+ b_{dk}e^{-2 \jmath\bar{\phi_m}}, ~~~ m = 1 , \ldots, M
\end{equation}
Since $|{\widetilde{r}_{k}}(\bar{\phi_m})= e^{-\jmath \bar{\phi_m}}| = |{\widetilde{r}_{k}}(\bar{\phi_m})|$,
 \begin{equation}
 \left|{\cal{F}}_{M}(\bzeta_{dk})\right| = \left[|\widetilde{r}_{k}(\bar{\phi}_1)|,| \widetilde{r}_{k}(\bar{\phi}_2)|, \ldots,| \widetilde{r}_{k}(\bar{\phi}_{M})|\right]^T.
 \end{equation}

\section{Proof of Lemma \ref{lemm_poly2}}\label{ap:rk1}
Let
\begin{equation}
\begin{aligned}
\Re\left\{r_k^*\frac{r_k^{(n)}}{\left|r_k^{(n)}\right|}\right\} = & \Re\left\{  \left[a_{dk}^*e^{-\jmath \phi_d} + {b_{dk}^* e^{\jmath \phi_d}+c_{dk}^*}\right] \frac{r_k^{(n)}}{\left|r_k^{(n)}\right|}\right\}  \\
 = &   \Re\left\{ \widetilde{a}_{dk}e^{-\jmath \phi_d} + {\widetilde{b}_{dk} e^{\jmath \phi_d}+\widetilde{c}_{dk}} \right\} \\
= &   (\widetilde{a}_{dkr}+\widetilde{b}_{dkr}) \cos(\phi_d) \\
& + (\widetilde{a}_{dki}-\widetilde{b}_{dki}) \sin{(\phi_d)} +\widetilde{c}_{dkr}
\end{aligned}
\end{equation}
where $\widetilde{a}_{dkr} = \Re({a_{dk}^{*}\frac{r_k^{(n)}}{\left|r_k^{(n)}\right|}})$, $\widetilde{b}_{dkr} = \Re({b_{dk}^{*}\frac{r_k^{(n)}}{\left|r_k^{(n)}\right|}})$, $\widetilde{c}_{dkr} = \Re({c_{dk}^{*}\frac{r_k^{(n)}}{\left|r_k^{(n)}\right|}})$, $\widetilde{a}_{dki} = \Im({a_{dk}^{*}\frac{r_k^{(n)}}{\left|r_k^{(n)}\right|}})$ and $\widetilde{b}_{dki} = \Im({b_{dk}^{*}\frac{r_k^{(n)}}{\left|r_k^{(n)}\right|}})$.\\
As in Appendix \ref{ap:rk2}, using the change of variable $\cos{(\phi_d)} = \frac{1-\beta_{d}^2}{1+\beta_{d}^2}$ and $\sin{(\phi_d)} = \frac{2\beta_{d}}{1+\beta_{d}^2}$,
\begin{equation}
\begin{aligned}
\Re\left\{r_k^*\frac{r_k^{(n)}}{\left|r_k^{(n)}\right|}\right\} = &  \frac{1}{(1+\beta_{d}^2)^2}\{ (\widetilde{a}_{dkr}+\widetilde{b}_{dkr})(1-\beta_{d}^4) \\
 & +(\widetilde{a}_{dki}-\widetilde{b}_{dki})(2\beta_{d})(1+\beta_d^2) \\
 & \widetilde{c}_{dkr}(1+\beta_{d}^2)^2. \}
  \end{aligned}
\end{equation}
Finally, defining $ \widetilde{\mu}_{dk} = \widetilde{c}_{dkr}- \widetilde{a}_{dkr}- \widetilde{b}_{dkr}$, $ \widetilde{\kappa}_{dk} = 2 (\widetilde{a}_{dki}-\widetilde{b}_{dki})$, $\widetilde{\xi}_{dk} = 2  \widetilde{c}_{dkr}$, $\widetilde{\eta}_{dk} =2 (\widetilde{a}_{dki}-\widetilde{b}_{dki})$ and
$ \widetilde{\rho}_{dk} = \widetilde{a}_{dkr} + \widetilde{b}_{dkr} +  \widetilde{c}_{dkr}$
yields

\begin{equation*}
\Re\left\{r_k^*\frac{r_k^{(n)}}{\left|r_k^{(n)}\right|}\right\}   =  \frac{ \widetilde{\mu}_{dk} \beta_{d}^4 + \widetilde{\kappa}_{dk} \beta_{d}^3 + \widetilde{\xi}_{dk} \beta_{d}^2 + \widetilde{\eta}_{dk} \beta_{d} + \widetilde{\rho}_{dk}}{(1+\beta_{d}^2)^2}.
\end{equation*}

\end{appendices}
\bibliographystyle{ieeetr}
\bibliography{ref1.bib}
\end{document}